%% file: main.tex
\documentclass[prx,twocolumn,superscriptaddress]{revtex4-2}
\usepackage[utf8]{inputenc}

\usepackage{amsmath}
\usepackage{amssymb}
\usepackage{amsfonts}
\usepackage{graphicx}
\usepackage[colorlinks=true, allcolors=blue]{hyperref}
\usepackage{braket}
\usepackage[noend]{algpseudocode}
\usepackage{algorithm}
\usepackage{bm}
\usepackage{qcircuit}
\usepackage{natbib}
\usepackage{amsthm}
\usepackage[normalem]{ulem}

\usepackage{comment}

\usepackage{tcolorbox}
\tcbuselibrary{breakable, skins, theorems}

\newcommand{\Tr}[1]{\mathrm{Tr}\left( #1 \right)}

\newcommand{\lstickx}[1]{\lstick{\makebox[1.5em][l]{$#1$}}}
\newcommand{\arrep}[1]{\ar @<4pt> @/^/[#1]|-{\mbox{ $\times L$ }}}

\newtheorem{theomain}{Theorem}
\newtheorem{corollary}{Corollary}
\usepackage{mwe,tikz}
\usepackage[percent]{overpic}

\begin{document}


\title{
\texorpdfstring{Optimal Parameter Configurations for Sequential Optimization of\\
Variational Quantum Eigensolver}{Optimal Parameter Configurations for Sequential Optimization of Variational Quantum Eigensolver}
}

\author{Katsuhiro Endo}
\affiliation{Research Center for Computational Design of Advanced Functional Materials, National Institute of Advanced Industrial Science and Technology (AIST), 1-1-1 Umezono, Tsukuba, Ibaraki, 305-8568, Japan}
\affiliation{Quantum Computing Center, Keio University, Hiyoshi 3-14-1, Kohoku-ku, Yokohama 223-8522, Japan}

\author{Yuki Sato}
\affiliation{Toyota Central R \& D Labs., Inc., Koraku Mori Building 10F, 1-4-14 Koraku, Bunkyo-ku, Tokyo 112-0004, Japan}
\affiliation{Quantum Computing Center, Keio University, Hiyoshi 3-14-1, Kohoku-ku, Yokohama 223-8522, Japan}

\author{Rudy~Raymond}
\affiliation{IBM Quantum, IBM Japan, 19-21 Nihonbashi Hakozaki-cho, Chuo-ku, Tokyo 103-8510, Japan}
\affiliation{Quantum Computing Center, Keio University, Hiyoshi 3-14-1, Kohoku-ku, Yokohama 223-8522, Japan}
\affiliation{Department of Computer Science, The University of Tokyo, 7-3-1, Hongo, Bunkyo-ku, Tokyo 113-0033, Japan}

\author{Kaito Wada}
\affiliation{Department of Applied Physics and Physico-Informatics, Keio University, Hiyoshi 3-14-1, Kohoku-ku, Yokohama 223-8522, Japan}

\author{Naoki Yamamoto}
\affiliation{Quantum Computing Center, Keio University, Hiyoshi 3-14-1, Kohoku-ku, Yokohama 223-8522, Japan}
\affiliation{Department of Applied Physics and Physico-Informatics, Keio University, Hiyoshi 3-14-1, Kohoku-ku, Yokohama 223-8522, Japan}


\author{Hiroshi C. Watanabe}
\affiliation{Quantum Computing Center, Keio University, Hiyoshi 3-14-1, Kohoku-ku, Yokohama 223-8522, Japan}

\begin{abstract}
Variational Quantum Eigensolver (VQE) is a hybrid algorithm for finding the minimum eigenvalue/vector of a given Hamiltonian by optimizing a parametrized quantum circuit (PQC) using a classical computer. 
Sequential optimization methods, which are often used in quantum circuit tensor networks, are popular for optimizing the parametrized gates of PQCs.  
This paper focuses on the case where the components to be optimized are single-qubit gates, in which the analytic optimization of a single-qubit gate is sequentially performed.
The analytical solution is given by diagonalization of a matrix whose elements are computed from the expectation values of observables specified by a set of predetermined parameters which we call the parameter configurations. 
In this study, we first show that the optimization accuracy significantly depends on the choice of parameter configurations due to the statistical errors in the expectation values. 
We then identify a metric that quantifies the optimization accuracy of a parameter configuration for all possible statistical errors, named configuration overhead/cost or \textit{C-cost}.
We theoretically provide the lower bound of C-cost and show that, for the minimum size of parameter configurations, the lower bound is achieved if and only if the parameter configuration satisfies the so-called equiangular line condition.
Finally, we provide numerical experiments demonstrating that the optimal parameter configuration exhibits the best result in several VQE problems.
We hope that this general statistical methodology will enhance the efficacy of sequential optimization of PQCs for solving practical problems with near-term quantum devices.
\end{abstract}

\maketitle

\section{Introduction}

Variational Quantum Eigensolver (VQE)~\cite{Peruzzo2014NatCom, Kandala2017Nat,TILLY20221} 
is a classical-quantum hybrid algorithm implementable on near-term quantum devices, 
for finding the minimum eigenvalue/vector of a given Hamiltonian; 
the recipe is simply to prepare a parametrized quantum circuit (PQC) $U(\bm{\theta})$, also called \textit{ansatz}, 
and then find a parameter $\bm{\theta}$ that minimizes 
$\langle H \rangle = 
\bra{\psi}U(\bm{\theta})^\dagger H U(\bm{\theta})\ket{\psi}$ with some initial 
state $\ket{\psi}$. 
Note that VQE is a class of the variational quantum algorithms (VQAs) 
\cite{cerezo2021variational,huang2022near}, where in general the cost is a non-linear 
function of the expectation values of some Hamiltonians. 
VQA has a wide range of applications, such as quantum chemical calculations~\cite{Peruzzo2014NatCom, Kandala2017Nat,gao2021applications}, combinatorial optimization~\cite{fuller2021approximate, amaro2022case, zoufal2022variational}, and linear equation solvers~\cite{bravo2019variational, xu2021variational, sato2021variational, sato2023variational}.

\begin{figure}[!b]
    \includegraphics[width=1.0\linewidth]{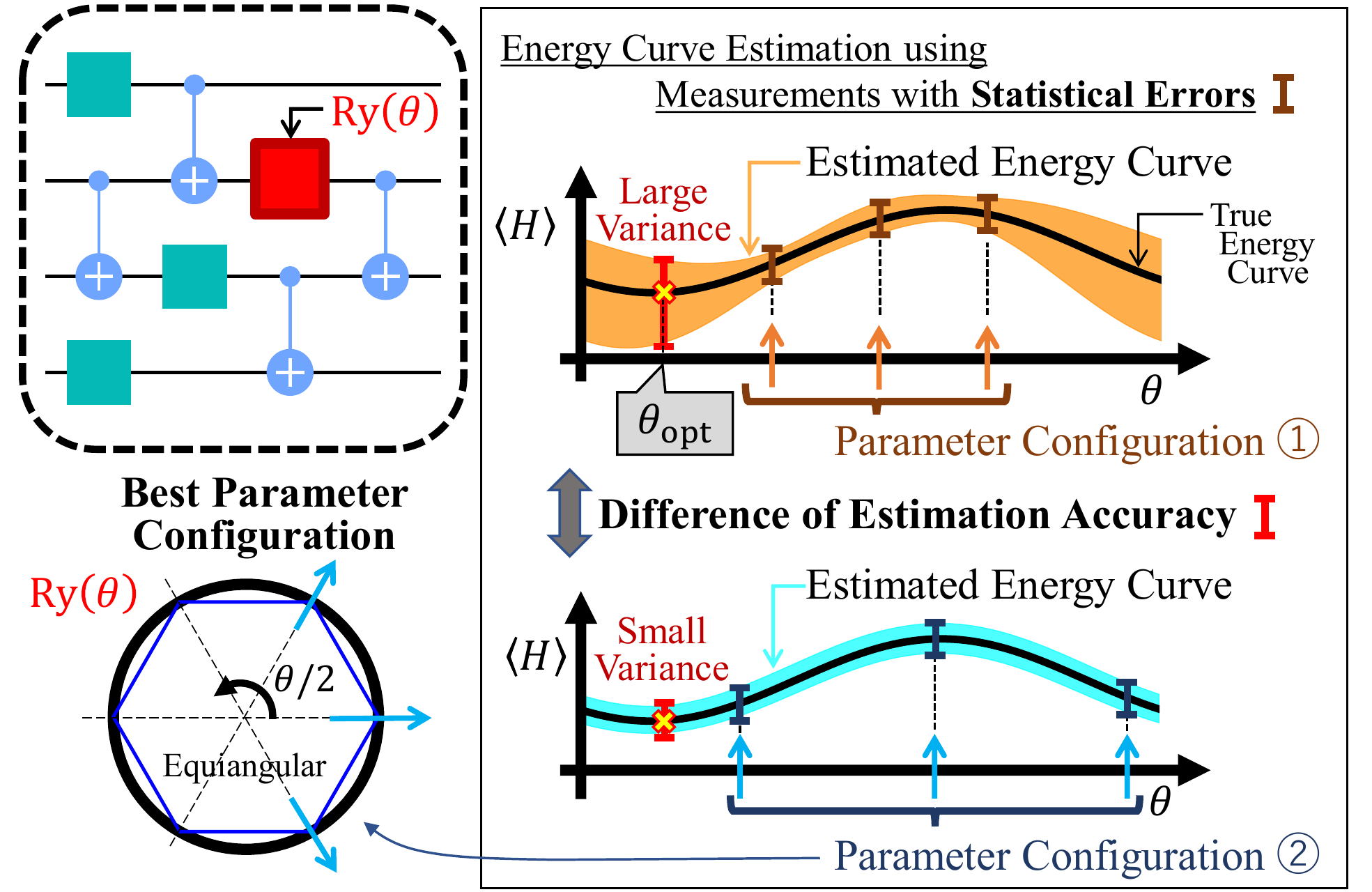}
    \caption{ \textbf{A general view to explain how the estimated optimal solution 
    varies depending on the parameter configuration, when there is statistical errors 
    for determining the cost function.} }
    \label{fig:figschem}
\end{figure}

The core question is how to model the PQC $U(\bm{\theta})$ and how to minimize 
$\langle H \rangle$ with some classical optimizer. 
There have been extensive investigation on this problem \cite{Cerezo2021NatRevPhys}. 
In particular, the {\it sequential optimization method} have been used in a variety 
of settings such as quantum circuit tensor-networks 
\cite{GarnetChan2022PRX,foss2021holographic,haghshenas2021optimization,barratt2021parallel,liu2019variational}, 
where $\bm{\theta}$ corresponds to a set of local unitaries and they are sequentually 
optimized one by one. 
In this paper, we focus on the special type of sequential optimization method developed 
in Refs.~\cite{nakanishi2020,ostaszewski2021,watanabe2021,wada2022,wada2022full}. 
In this framework, $\bm{\theta}$ are the parameters characterizing the set of single-qubit 
rotation gates such as $R_y(\theta) = e^{i\theta Y}$ ($Y$ is the Pauli $y$ 
matrix) in the case of {\it Rotosolve} \cite{nakanishi2020,ostaszewski2021}. 
Then the sequential optimization method takes the strategy to {\it exactly} optimize 
the single rotation gates one by one. 
For example, consider the step where we optimize the $R_y(\theta)$ gate contained 
in the PQC shown in Fig.~\ref{fig:figschem} by minimizing the cost $\langle H \rangle$ 
as a function of $\theta$. 
The point is that, in this case, $\langle H \rangle$ must be of the form of a sinusoidal 
function with respect to $\theta$, and thus the optimal $\theta_{\rm opt}$ can be exactly 
determined once we identify the sinusoidal function shown by the black curve in the figure. 
In particular, as a nature of sinusoidal function, specifying the mean values of three 
observables corresponding to the three points of $\theta$ allows us to exactly identify 
$\langle H \rangle$; 
we call the alignment of these three points of $\theta$ the {\it parameter configuration}. 
Note that, in the case of {\it Free axis selection (Fraxis)} \cite{watanabe2021} where the freedom of a single-qubit 
rotation gate is served by the rotation axis with fixed rotation angle in the Bloch sphere, 
$\langle H \rangle$ takes the form of a quadratic function of a real normalized vector 
$\bm{n}=(x,y,z)^T$, which can also be exactly minimized. 
This setup was further generalized to {\it Free Quaternion Selection (FQS)} 
\cite{wada2022,wada2022full} so that the rotation angle can also be tuned; 
then $\langle H \rangle$ takes the form of a quadratic function of a real normalized 
vector $\bm{q}=(w,x,y,z)^T$. 
In this case, as shown later, the mean values of 10 observables corresponding to 10 points 
of $\bm{q}$ identify $\langle H \rangle$; we also call this $\{\bm{q}_1, \ldots, \bm{q}_{10} \}$ 
the parameter configuration.

However, this optimization strategy relies on the critical assumption that the mean 
values of observables and accordingly $\langle H \rangle$ are exactly identified. 
In reality, those mean values can only be approximately obtained as the average of 
a finite number of measurement results; that is, practically there is always a 
statistical error in $\langle H \rangle$. 
In the above one-dimensional case, as illustrated in Fig.~\ref{fig:figschem}, the energy 
curve, $\theta_\mathrm{opt}$, and consequently the minimum value of $\langle H \rangle$ may 
all largely fluctuate depending on the parameter configuration. 
Hence the question is what is the best parameter configuration for achieving a small 
fluctuation of ${\rm min}\langle H \rangle$. 
In the above one-dimensional case, we have an intuition that the best configuration might 
be such that the three parameters are equally spaced (i.e., equidistant), as shown in the 
left bottom of Fig.~\ref{fig:figschem}, which is indeed true as proven later. 
However, the general case is of course nontrivial; will we have such equidistant 
configuration in some sense, or some biased configuration would be the best? 

\begin{figure*}[t]
    \includegraphics[width=1.0\linewidth]{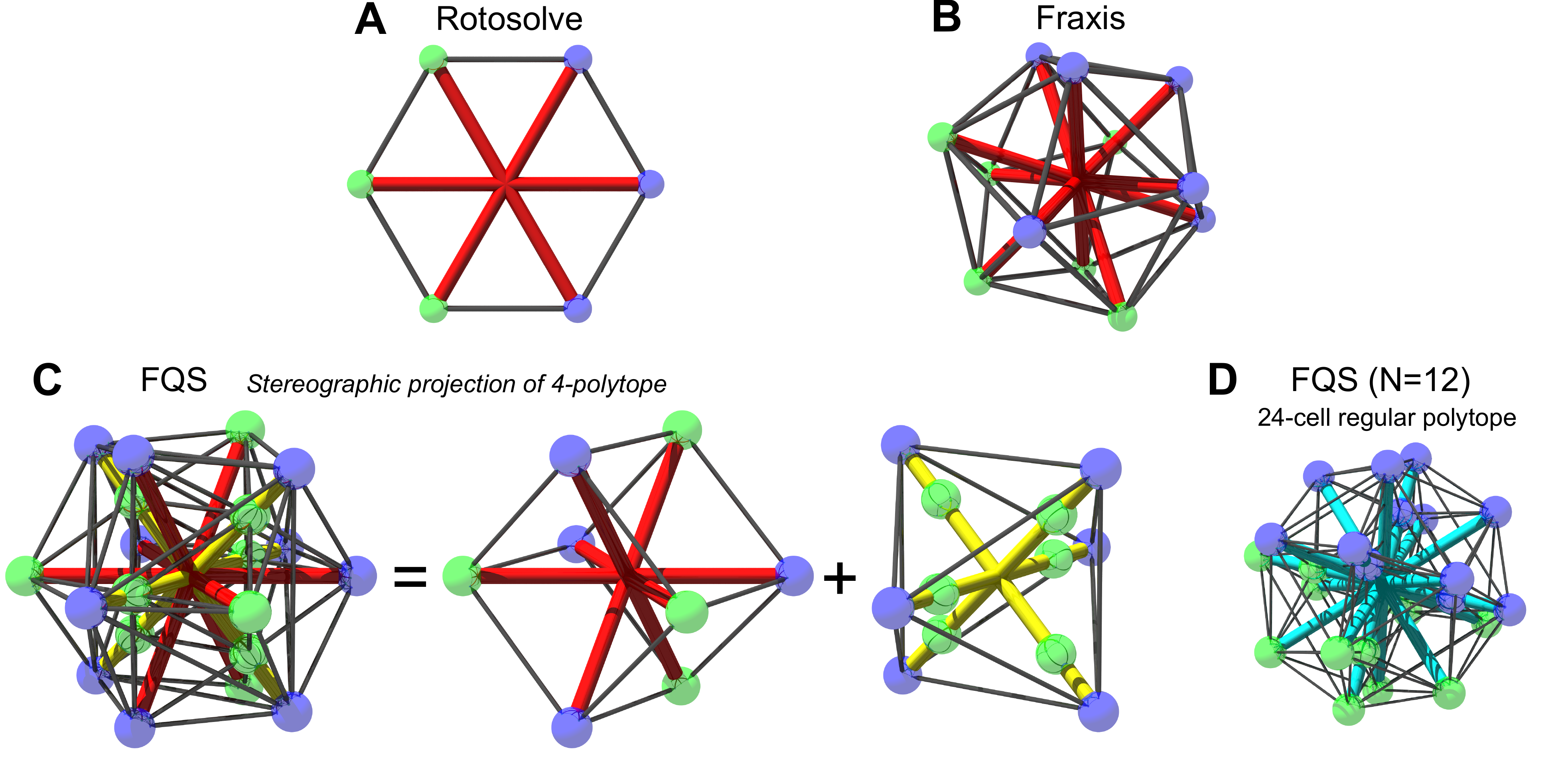}
    \caption{ \textbf{
    Optimal parameter configurations for Rotosolve, Fraxis, and FQS models.} 
    Blue spheres represent the optimal configuration of $\{\bm{q}_i\}_{i=1}^{N}$, and green 
    spheres represent its opposite position, $\{-\bm{q}_i\}_{i=1}^N$. 
    (A) Rotosolve: The diagonal lines between $\bm{q}_i$ and $-\bm{q}_i$ constitute three 
    equiangular lines in 2D space (red lines). 
    (B) Fraxis: The diagonal lines constitute six equiangular lines in 3D space. 
    (C) FQS: The four positions and its opposite in the first term of the right hand 
    side constitute a regular cube in a hyperplane; the six positions in the second 
    term (outer blues) constitute a regular octahedron in a hyperplane. 
    The opposite of last six positions (inner greens) also constitute a regular 
    octahedron. 
    The yellow diagonals are apparently doubly overlapped due to the stereographic projection, but they are not actually overlapped. 
    (D) FQS ($N$=12): The 24-cell polytope in the 4-dimensional space, which achieves $C(A)=1$ as a redundant parameter configurations.
    }
    \label{fig:optimalMCs}
\end{figure*}

In this paper, we develop the theory for determining the optimal parameter configuration. 
As a preliminary result, in Sec.~\ref{sec:derivation}, we prove that, if the exact 
expectation values are available without any statistical error, then we have analytical 
solution of the best parameters achieving ${\rm min}\langle H \rangle$ (almost) without 
respect to the parameter configuration for every method of 
\cite{nakanishi2020,ostaszewski2021,watanabe2021,wada2022,wada2022full}. 
Then, in Sec.~\ref{sec:C-cost}, we give the most essential result providing the basis 
of the theory; that is, we derive the explicit form of the fluctuation of 
${\rm min}\langle H \rangle$ under statistical errors, with respect to the parameter 
configuration. 
This enables us to introduce the {\it C-cost} (configuration cost/overhead), a useful metric for 
determining ${\rm min}\langle H \rangle$ and thereby providing us with the optimal 
parameter configuration. 
Actually, Sec.~\ref{sec:experiments} gives numerical experiments to demonstrate that 
the optimal parameter configurations obtained using C-cost yield the best result in 
the sense of the statistical error of estimating $\langle H \rangle$.

Notably, beyond such utilization for numerically determining the configuration, the 
C-cost satisfies several interesting mathematical properties, suggesting the relevance 
of this metric. 
The first is that the lower bound of C-cost is 1; moreover, we prove that, for the 
minimum size of the parameter set, this bound is achievable if and only if the parameter 
configuration satisfies a geometric condition called the {\it equiangular line condition}, 
an important and beautiful mathematical concept in algebraic graph theory. Here, each parameter $\bm{q}$ corresponds to a line that passes the origin and $\bm{q}$.
This condition rigorously supports our above-described intuition that it would be desirable 
for the parameters to be equally spaced for the Rotosolve case shown in 
Fig.~\ref{fig:figschem} or Fig.~\ref{fig:optimalMCs}A; 
this intuition holds for the case of Fraxis, showing that there is a unique parameter 
configuration (up to the global rotation) satisfying the equiangular line condition, 
as displayed in Fig.~\ref{fig:optimalMCs}B. 
But interestingly, this intuition does not apply to the most general FQS case due to the 
non-existence of 10 equiangular lines in $\mathbb{R}^4$. 
That is, the so-called Gerzon bounds~\cite{LemmensSeidel73}, Neumann theorem~\cite{LinYu2020} 
and Haantjes bound~\cite{haantjes1948equilateral} prove that 
there does not exist a set of 10 lines satisfying the equiangular line condition in $\mathbb{R}^4$; the maximum number of such lines is 6. 
Nevertheless, the C-cost is still useful in this case, as it gives us a means to numerically obtain 
the optimal parameter configuration, which is displayed in Fig.~\ref{fig:optimalMCs}C. 
Furthermore, if redundant measurements are allowed, there exists 
parameter configurations that achieves the theoretical lower bound of the C-cost, 
one of which is illustrated in Fig.~\ref{fig:optimalMCs}D.

Finally, we note that equiangular lines in complex spaces are equivalent to symmetric, informationally complete (SIC) POVMs~\cite{Renes2004} whose properties have been much studied, e.g., it is conjectured that there is always a set of $d^2$ equiangular lines in $\mathbb{C}^d$~\cite{Scott2010} (it has been proven up to some large $d$ theoretically and numerically). The SIC POVMs defined from such lines are \textit{informationally complete} because the results of other measurements can be computed from those of the SIC POVMs. In this study, we obtain similar results connecting equiangular lines in real spaces with the variational quantum circuits using parametrized single-qubit gates.


\section{Energy minimization with matrix factorization} 
\label{sec:derivation}

\subsection{Brief review of Rotosolve, Fraxis, and FQS}

FQS method~\cite{wada2022full} describes the procedure to completely characterize the energy landscape with respect to a single-qubit gate in a PQC.
The parametrized single-qubit gate, which we call {\it FQS gate}, is none other than the general single-qubit gate $U^{(4)}$ expressed as~\cite{wada2022full,wharton2015unit}
\begin{equation}
\label{eq:gengate}
    U^{(4)}(\bm{q}) = wI - xiX - yiY - ziZ =\bm{q}\cdot \vec{\varsigma},
\end{equation}
where the superscipt indicates the number of parameters: $\bm{q}=(w,x,y,z)^T\in \mathbb{R}^4$ satisfying $\|\bm{q}\|^2=1$.
Here, $i$ is the imaginary unit, $I$ is the 2$\times$2 identity matrix, and $X, Y, Z$ are the Pauli matrices.
$\vec{\varsigma}=(\varsigma_I,\varsigma_X,\varsigma_Y,\varsigma_Z)^T$ denotes an extension of the Pauli matrices defined as
\begin{align}
    \vec{\varsigma}=(I,-iX,-iY,-iZ)^T.
\end{align}
The dimension of the parameter $\bm{q}$ is four, but since the parameter $\bm{q}$ is constrained on the unit hyper-sphere, the degree of freedom of the parameter is three.


In Fraxis, the rotation angle is constrained to $\pi$, which corresponds to the case 
$w=0$ of Eq.~\eqref{eq:gengate} as
\begin{equation}
    U^{(3)}(\bm{n}) = -xiX -yiY - ziZ,
\end{equation}
where the parameter of the gate is $\bm{n}=(x,y,z)^T$ such that $\|\bm{n}\|^2=1$. 
We term this $U^{(3)}$ as {\it Fraxis gate}. 
Thus, the Fraxis gate has two degrees of freedom.

In Rotosolve, the rotation axis is fixed and the rotation angle serves as the parameter.
In particular, {\it Rx gate} fixes the rotation axis to the $x$-axis; in the form 
of Eq.~\eqref{eq:gengate}, this corresponds to $y=z=0$ and thus
\begin{equation}
    \label{eq:gengateRx}
    U^{(2)}(\bm{r}) = wI - xiX,
\end{equation}
where the parameter of the gate is $\bm{r}=(w,x)^T$ such that $\|\bm{r}\|^2=1$.
Thus, the degree of freedom of Rx gate is one. 
Similarly, Ry and Rz gates are obtained by replacing $X$ in Eq.~\eqref{eq:gengateRx} 
with $Y$ and $Z$, respectively.

In what follows we use the most general FQS gate to describe the optimization 
algorithm. 
The sequential optimization method takes the strategy to update respective FQS gates 
in a coordinate-wise manner, where all parameters are fixed except for the focused FQS 
gate $U^{(4)}({\bm q})$. 
The entire quantum circuit containing FQS gates is supposed to be the PQC 
$V=\prod_i U^{(4)}_i({\bm q}_i) W_i $ on the $n$-qubit system, where $U^{(4)}_i$ is the $i$th FQS  gate and $W_i$ is a fixed multi-qubit gate.

Now, let $V_1$ and $V_2$ be the gates placed before and after the focused FQS gate $U^{(4)}({\bm q})$.
Then, a density matrix $\rho$ prepared by the PQC is expressed as 
\begin{equation}
    \rho = V_{2} U^{(4)}(\bm{q}) V_{1} \rho_{\rm in} V_{1}^\dag \left(U^{(4)}(\bm{q}) \right)^\dag V_{2}^\dag,
\end{equation}
where $\rho_{\rm in}$ is an input density matrix.
Thus, the expectation value $\langle H \rangle$ of given Hamiltonian $H$ with respect 
to $\rho$ is then 
\begin{eqnarray}
    \label{eq:eofham}
    \langle H \rangle 
    &=& \Tr{HV_{2} U^{(4)}(\bm{q}) V_{1}\rho_{\rm in} V_{1}^\dag \left(U^{(4)}(\bm{q})\right)^\dag V_{2}^\dag} \notag \\
    &=&  \Tr{H' U^{(4)}(\bm{q}) \rho'_{\rm in} \left(U^{(4)}(\bm{q})\right)^\dag},
\end{eqnarray}
where 
$H' = V_{2}^\dag H V_{2}$ and $\rho'_{\rm in} = V_1 \rho_{\rm in} V_1^\dag$. 
Substituting Eq.~\eqref{eq:gengate} into Eq.~\eqref{eq:eofham} yields
\begin{align}\label{eq:hmat}
    \braket{H}=\bm{q}^T G^{(4)}\bm{q},
\end{align}
where $G^{(4)}$ is a $4\times 4$ real-symmetric matrix:
\begin{equation}\label{eq:Gmat}
    G^{(4)} = \begin{bmatrix}
    G_{II} & G_{IX} & G_{IY} & G_{IZ} \\
    G_{IX} & G_{XX} & G_{XY} & G_{XZ} \\
    G_{IY} & G_{XY} & G_{YY} & G_{YZ} \\
    G_{IZ} & G_{XZ} & G_{YZ} & G_{ZZ}
    \end{bmatrix},
\end{equation}
and each element, $G_{\mu\nu}~(\mu,\nu=I,X,Y,Z)$, is defined by
\begin{align}
    G_{\mu\nu}=\frac{1}{2}\Tr{\rho'_{\rm in}\left(\varsigma_{\mu}^\dagger H'\varsigma_{\nu}+\varsigma_{\nu}^\dagger H'\varsigma_{\mu}\right)}.
\end{align}
Thus the energy landscape with respect to the FQS gate is completely characterized by the matrix $G^{(4)}$.
Because Eq.~\eqref{eq:hmat} is a quadratic form with respect to the parameter $\bm{q}$ 
with the constraint $\|\bm{q}\|^2=1$, the eigenvector 
$\bm{p}_1$ associated with the lowest eigenvalue $\lambda_1$ of the matrix $G^{(4)}$ minimizes the energy \eqref{eq:hmat}; see Appendix~{\ref{apdx:FQS}} for the details.
In the following, we call the matrix $G^{(4)}$ \textit{FQS matrix}. 
Note that the above result can be directly extended to the case of Fraxis and Rotosolve, in which case Eq.~\eqref{eq:Gmat} is replaced by 
\begin{equation}\label{eq:gfraxis}
    G^{(3)} = \begin{bmatrix}
    G_{XX} & G_{XY} & G_{XZ} \\
    G_{XY} & G_{YY} & G_{YZ} \\
    G_{XZ} & G_{YZ} & G_{ZZ}
    \end{bmatrix},
\end{equation}
and
\begin{equation}\label{eq:groto}
    G^{(2)} = \begin{bmatrix}
    G_{II} & G_{IX}\\
    G_{IX} & G_{XX}
    \end{bmatrix},
\end{equation}
respectively.

\subsection{FQS with arbitrary parameter configurations}
\label{subsec:ourproposal}

Since $G^{(4)}$ is a real-symmetric matrix, we can expand Eq.~\eqref{eq:hmat} as the following form:
\begin{align}
\label{eq:htenkai}
\begin{split}
    \langle H \rangle  
    &= G_{II} w^2
    + G_{XX} x^2
    + G_{YY} y^2
    + G_{ZZ} z^2 \\
    &+ 2 G_{IX} wx
    + 2 G_{IY} wy
    + 2 G_{IZ} wz \\
    &+ 2 G_{XY} xy
    + 2 G_{XZ} xz
    + 2 G_{YZ} yz.
\end{split}
\end{align}
Eq.~\eqref{eq:htenkai} indicates that, if we know all the 10 coefficients  
$(G_{II},...,G_{YZ})$, we can exactly estimate the expectation $\langle H \rangle $ 
for any parameter $\bm{q}$. 
In other words, only algebraic calculations on classical computers are required to find the parameters achieving the minimum expectation value for the target gate.

Therefore, it is important to obtain the coefficients with as few measurements as 
possible. 
To consider this problem, we define the function $h^{(4)}(\bm{q})$ that outputs the 
normalized vector ($\|\bm{h}^{(4)}(\bm{q})\|=1$): 
\begin{align}
    \label{eq:hfunc}
    &\bm{h}^{(4)}(\bm{q}) \notag\\
    &~=\!
    (w^2\!\!, x^2\!\!, y^2\!\!, z^2\!\!, \!\sqrt{\!2}wx, \!\sqrt{\!2}wy, \!\sqrt{\!2}wz, \!\sqrt{\!2}xy, \!\sqrt{\!2}xz, \!\sqrt{\!2}yz)^{\!T} \!\!\!,
\end{align}
and the vector $\bm{g}^{(4)}$
\begin{align}
    \label{eq:gvec}
    \bm{g}^{(4)}&=(
    G_{II},G_{XX},G_{YY},G_{ZZ}, \notag \\&
    \!\sqrt{\!2}G_{IX}, 
    \!\sqrt{\!2}G_{IY},
    \!\sqrt{\!2}G_{IZ},
    \!\sqrt{\!2}G_{XY},
    \!\sqrt{\!2}G_{XZ},
    \!\sqrt{\!2}G_{YZ})^T.
\end{align}
Then, the relation between the parameter $\bm{q}$ and the expectation 
$\langle H \rangle$ is expressed as 
\begin{equation}
    \label{eq:H_hq}
    \langle H \rangle = \bm{h}^{(4)}(\bm{q})^T ~ \bm{g}^{(4)}.
\end{equation}

Suppose measurements with different parameters $\{\bm{q}_1, ..., \bm{q}_N\}$ and the $N$ expectation values of the measurement results $\bm{b}=(b_1, ..., b_N)^T$ were obtained,
we can also write the relations between the expectation values $\bm{b}$ and the coefficient vector $\bm{g}^{(4)}$ as 
\begin{equation}
    \label{eq:bAg}
    \bm{b} = A^{(4)}\bm{g}^{(4)},
\end{equation}
where the matrix $A^{(4)} \in \mathbb{R}^{N \times 10}$ is
\begin{equation} \label{eq:A-matrix}
    A^{(4)} = ( \bm{h}^{(4)}(\bm{q}_1), ..., \bm{h}^{(4)}(\bm{q}_N) )^T,
\end{equation}
that encodes the information of the parameter configurations $\{\bm{q}_1, ..., \bm{q}_N\}$.

It is obvious, if $N<10$, $\bm{g}^{(4)}$ is not uniquely determined.
Hence, we suppose $N\geq 10$ throughout this paper.
If ${\rm rank}(A)=10$, $A^TA$ is invertible and there exists the generalized inverse of  $A^+:=(A^TA)^{-1}A^T$  \cite{penrose1955generalized}. 
Accordingly, we can obtain the vector $\bm{g}^{(4)}$ by exactly solving linear equations as
\begin{equation}
    \label{eq:gAinvb}
    \bm{g}^{(4)} = A^{+}\bm{b}.
\end{equation}
In other words, a single execution of FQS requires at least ten sets of the parameters and the corresponding observables.
However, it may not necessarily be the case when input states and/or Hamiltonian has symmetry, which reduces the number of required measurements to construct $G^{(4)}$ in Eq.~(\ref{eq:Gmat}).
We also note that it is possible that ${\rm rank}(A)<10$ if the rows of $A$ are dependent on each other.
However, it is plausible to exclude such situation, because the input parameters are controllable.
Hereafter, we suppose that all columns of $A$ are independent of each other, equivalently, ${\rm rank}(A)=10$. 

The same argument is applicable to the Fraxis gate as
\begin{eqnarray}
    \langle H \rangle  
    &\!\!=\!\!& G_{XX} x^2
    + G_{YY} y^2
    + G_{ZZ} z^2 \nonumber \\
    && + 2 G_{XY} xy
    + 2 G_{XZ} xz
    + 2 G_{YZ} yz,
 \\
    \bm{h}^{(3)}(\bm{n})
    &\!\!=\!\!& (x^2, y^2, z^2, \!\sqrt{\!2}xy, \!\sqrt{\!2}xz, \!\sqrt{\!2}yz)^T, \\
    \bm{g}^{(3)} &\!\!=\!\!& (G_{XX},G_{YY}, G_{ZZ},\notag\\
    &&~~~~~\!\sqrt{\!2}G_{XY}, \!\sqrt{\!2}G_{XZ}, \!\sqrt{\!2}G_{YZ})^T. \ \ \ \ \ 
\end{eqnarray}
Likewise, for Rx gates
\begin{eqnarray}
    \langle H \rangle  
    &=& G_{II} w^2
    + G_{XX} x^2
    + 2 G_{IX} wx, \\
    \bm{h}^{(2)}(\bm{r})
    &=& (w^2, x^2, \!\sqrt{\!2}wx)^T,\\
    \bm{g}^{(2)}
    &=& (G_{II}, G_{XX}, \!\sqrt{\!2}G_{IX})^T.
\end{eqnarray} 
The minimum sizes of the parameter configuration required to construct $G^{(d)}$ are $d(d+1)/2$, i.e., 6 in Fraxis ($d=3$) and 3 in Rotosolve ($d=2$).
For simplicity, we omit superscript $d$ from $G^{(d)}$, $h^{(d)}$, and $\bm{g}^{(d)}$ for $d=2,3,4$ in the following sections and formulate them based on the FQS framework unless otherwise noted.

\section{Configuration cost with finite runs of quantum circuits}\label{sec:C-cost}

\subsection{Evaluation of the Parameter Configurations}\label{sec:param_config}
If infinite number of measurements were allowed, there would be no estimation errors in the expectation values $\bm{b}$, and the resulting vector $\bm{g}$ is exactly obtained as long as the matrix $A$ is invertible. This allows for the exact evaluation of the optimal solution of the FQS matrix.
In this section, we quantitatively evaluate the error propagation from the shot noise in the expectation values $\bm{b}$ to the estimation of the minimum solution. 
Although we focus on the FQS for generality, it can be easily applied to other sequential quantum optimizers, Rotosolve and Fraxis.  
Suppose a FQS matrix is estimated from $N$ expectation values of an observable, which are obtained by independent measurements with different parameters $\{\bm{q}_1, ..., \bm{q}_N\}$ assigned to the gate of interest.
Due to the finite number of shots, the expectation values are no longer deterministic, but randomly distribute around the true values $\bm{b}^*$ obtained with infinite shots as 
\begin{equation}
    \label{eq:bpluse}
    \bm{b} = \bm{b}^* + \bm{\epsilon},
\end{equation}
where $\bm{\epsilon}$ is the random variables reflecting the errors on the measurements.

Note that the relation between $\bm{b}$ and  $\bm{g}$ is no longer valid under the finite measurement condition.
Alternatively, we employed the least-square solution $\bm{g}$
\begin{align}
    \label{eq:glsq}
    \bm{g} = \underset{\tilde{\bm{g}}}{\arg\min} \|\bm{b}-A\tilde{\bm{g}}\|^2 = (A^TA)^{-1}A^T \bm{b} = A^+\bm{b},
\end{align}
as a plausible estimate of $\bm{g}^*$.
Apparently, Eq.~\eqref{eq:glsq} has the same form as Eq.~\eqref{eq:gAinvb}, but the resulting vector $\bm{g}$ is an estimate of the true vector $\bm{g}^*$ in the context of maximum likelihood~\cite{don1985use} and deviates from the ideal vector $\bm{g}^*$ due to errors for finite measurement.
Substituting Eq.~\eqref{eq:bpluse} into Eq.~\eqref{eq:glsq}, we get
\begin{align}
    \label{eq:g_ae}
\bm{g}&= A^+ \bm{b}\notag \\
&=A^+(\bm{b}^*+\bm{\epsilon}) \notag \\
    &=\bm{g}^{*} + A^+\bm{\epsilon},
\end{align}
where the third equality follows $\bm{g}^{*}=A^{+}\bm{b}^{*}$. 
Eq.~\eqref{eq:g_ae} implies the errors of the estimated coefficient vector $\bm{g}-\bm{g}^*=A^+\bm{\epsilon}$ is amplified by the linear transformation $A^+$ from the shot errors $\bm{\epsilon}$.

Let $G$ be a FQS matrix generated from the estimated vector $\bm{g}$ with finite number of measurements. 
In the below, we focus on the FQS procedure to estimate the minimum eigenvalue of $G$.
Here, for convenience, we define the half-vectorization function $\mathrm{vech}: {\mathbb R}^{4\times4} \rightarrow {\mathbb R^{10}}$
such that
\begin{align}
    &\mathrm{vech}(G)~\!\!=\!\! \notag \\
    & \ \ (
    G_{\!I\!I}\!,
    G_{\!X\!X}\!,
    G_{\!Y\!Y}\!,
    G_{\!Z\!Z}\!,
    G_{\!I\!X}\!, 
    G_{\!I\!Y}\!,
    G_{\!I\!Z}\!,
    G_{\!X\!Y}\!,
    G_{\!X\!Z}\!,
    G_{\!Y\!Z})^T,
\end{align}
where the order of elements corresponds to $\bm{g}$.
In addition, the scaling matrix $D$ is defined as
\begin{align}
    D = \mathrm{diag}(1,1,1,1,\sqrt{\!2},\sqrt{\!2},\sqrt{\!2},\sqrt{\!2},\sqrt{\!2},\sqrt{\!2}).
\end{align}
Using these notations,
we have the following relations,
\begin{equation}
\bm{g}=D~\mathrm{vech}(G)
\Leftrightarrow 
G=\mathrm{vech}^{-1}(D^{-1}\bm{g}),
\end{equation}
where the function $\mathrm{vech}^{-1}$ is a linear mapping as $\mathrm{vech}^{-1}(\bm{s}+\bm{t})=\mathrm{vech}^{-1}(\bm{s})+\mathrm{vech}^{-1}(\bm{t})$ for $\bm{s,t}\in \mathbb{R}^{10}$.
Accordingly, $G$ is expressed as 
\begin{equation}
    \label{eq:Gerr}
    G = \mathrm{vech}^{-1}(D^{-1}\bm{g}) 
    = G^* +\mathrm{vech}^{-1}(D^{-1}A^+\bm{\epsilon}),
\end{equation}
which implies that the  ideal FQS matrix $G^*={\rm vech}^{-1}(D^{-1}A^{+}\bm{b}^*)$ is perturbed by $\mathrm{vech}^{-1}(D^{-1}A^+\bm{\epsilon})$.

In the following part, we quantitatively evaluate the matrix perturbation effect on the optimization result.
Let $\lambda_i^*$ and ${\bm p}_i^*$ be the $i$th lowest eigenvalue and the corresponding eigenvector of $G^*$. 
Likewise, $\lambda_i(\bm{\epsilon})$ and $\bm{p}_i(\bm{\epsilon})$ are the $i$th lowest eigenvalue and its corresponding eigenvector of the estimated matrix $G$.
For quantitative evaluation of the perturbation, we suppose two metrics: (1) $\mathrm{Var}[\lambda_1(\bm{\epsilon})]$, the variance of the estimated minimum value, and (2) $\mathbb{E}[\Delta E]$, the mean error of the minimum expectation value using the estimated optimal parameters with infinite shot.
Here, $\Delta E$ is the deviation of the expectation value with the estimated parameter set $\bm{p}_1$ from the true minimum expectation value, defined as
\begin{equation}
\label{eq:EnergyDeviation}
    \Delta E = \bm{p_1}^{T} G^*  \bm{p}_1^{} - \bm{p_1}^{*T} G^*  \bm{p_1}^{*} \geq 0,
\end{equation}
where the positivity of $\Delta E$ comes from the fact that the true parameter set ${\bm p}^*_1$ gives the minimum value of the quadratic form.
We suppose that $\mathrm{Var}[\lambda_1(\bm{\epsilon})]$ is a measure to verify the estimated energy $\lambda_1$ by one-time execution of FQS, while $\mathbb{E}[\Delta E]$ is a measure to qualify the estimated parameter ${\bm{p}_1}$.
Throughout the following parts, for simplicity, we employed $\mathrm{Var}[\lambda_1]$ as the indicator of shot errors.
(See Appendix~\ref{apdx:alternative_measure} for $\mathbb{E}[\Delta E ]$)

Since $G$ is a $4\times 4$ symmetric matrix, it is represented by eigendecomposition as
\begin{equation}
    \label{eq:Gmodel}
    G=P \Lambda P^T,
\end{equation}
where $P=(\bm{p}_1,...,\bm{p}_4)^T$ and $\Lambda = \mathrm{diag}(\lambda_1,...,\lambda_4)$. 
From the first-order perturbation theory \cite{kato2013perturbation}, the minimum eigenvalue $\lambda_1$ of $G$ is approximated as
\begin{equation}
    \lambda_1 = \lambda_1^* + \bm{p}_1^{*T} \mathrm{vech}^{-1}(D^{-1}A^+\bm{\epsilon}) \bm{p}^*_1.
\end{equation}
Then, $\mathrm{Var}[\lambda_1]$ is evaluated as 
\begin{equation}
\label{eq:score_var}
    \mathrm{Var}[\lambda_1] = \mathrm{Var}[\bm{p}_1^{*T} \mathrm{vech}^{-1}(D^{-1}A^+\bm{\epsilon}) \bm{p}^{*}_1].
\end{equation}

To deal with Eq.~\eqref{eq:score_var}, we apply a simple model to the measurement errors $\bm{\epsilon}$ satisfying as
\begin{equation}
\mathbb{E}[\bm{\epsilon}]=\bm{0},
\end{equation}
\begin{equation}
    \mathbb{E} [\epsilon_i \epsilon_j] =
    \left\{
    \begin{array}{ll}
    0 & \mathrm{for}~i\neq j\\
    \sigma^2/s & \mathrm{for}~i=j
    \end{array}
    \right. ,
\end{equation}
where $s$ denotes the number of measurement shots to evaluate an expectation value of observables and $\sigma^2$ is a part to specific to observables. 

In addition, we assume the first eigenvector $\bm{p}_1$ follows a uniform distribution on the unit sphere.
Based on the models, Eq.~\eqref{eq:score_var} can be further calculated as
\begin{equation}
\label{eq:score_ksk}
    \mathrm{Var}[\lambda_1] = \frac{\sigma^2}{sd(d+2)}
    {\rm Tr}[ (A^TA)^{-1}(\bm{1}_d\bm{1}_d^T+2I)],
\end{equation}
where $d=\mathrm{dim}(\bm{q})$ (4 for FQS, 3 for Fraxis and 2 for Rx) and $\bm{1}_d\in \mathbb{R}^{d(d+1)/2}$ is the vector that the first $d$ elements are unity and the others are zero (e.g. $\bm{1}_d=(1,1,1,1,0,0,0,0,0,0)^T$ for FQS).
Derivation of Eq.~\eqref{eq:score_ksk} is detailed in Appendix~\ref{apdx:analytical_form}.

Since we focus on the optimization performance,
it is convenient to discuss the total number of shots required for an one-time optimization rather than the cost for evaluating an expectation value. 
Suppose the total shots for an one-time optimization is constant.
Let ${s_\mathrm{min}}$ be the number of measurement shots to estimate an expectation value of the observable when $N=N_\mathrm{min}$, where $N_\mathrm{min}:=d(d+1)/2$ is the minimum size of the parameter configuration.
For a redundant parameter configuration $N>N_\mathrm{min}$, the number of shots for evaluating an expectation value is $s_{\mathrm{min}}N_{\mathrm{min}}/N$.
As a result,
\begin{eqnarray}
\label{eq:score_ksk2}
    \mathrm{Var}[\lambda_1] 
    = \frac{\sigma^2}{s_\mathrm{min}} C(A),
\end{eqnarray}
where we define the C-cost (Configuration cost), $C(A)$, as
\begin{equation} \label{eq:ecost}
    C(A):= \frac{N}{N_{\mathrm{min}}d(d+2)}
    {\rm Tr}[ (A^TA)^{-1}(\bm{1}_d\bm{1}_d^T+2I)].
\end{equation}

 Equation~\eqref{eq:score_ksk2} indicates that $\mathrm{Var}[\lambda_1]$ is separable into the number of shots ($s_\mathrm{min})$ dependent part and the parameter configuration dependent part i.e. a 50\% reduction of $C(A)$ is equivalent to doubling the number of shots.
The C-cost is a metric to estimate $\mathrm{Var}[\lambda_1]$ under the condition that the number of shots to optimize a single-qubit gate is constant.

Now, the conditions for the minimum $C(A)$ are of interest to minimize the estimation error. We rigorously give the lower bound of the C-cost as the following theorem (See Appendix~\ref{sec:ecost_theorems} for the proof of this theorem):
\begin{theomain}\label{theorem:extended_C-cost}
For the C-cost $C(A)$ in Eq.~\eqref{eq:ecost}, $C(A) \geq 1$ holds with equality if and only if the parameter configurations $\{ \bm{q}_i \}_{i=1}^{N}$ satisfy 
    \begin{equation} \label{eq:condOfMinExEcost}
        A^TA = \frac{N}{d(d+2)} (\bm{1}_d\bm{1}_d^T + 2I).
    \end{equation}
\end{theomain}
In other words, the parameter configurations that satisfies Eq.~\eqref{eq:condOfMinExEcost} is optimal with respect to efficiency.
Although it may not be straightforward to find the optimal parameter sets that satisfy Eq.~\eqref{eq:condOfMinExEcost},  in the case of minimum parameter set ($N=N_\mathrm{min}$) a useful formula is available as the following corollary of Theorem~\ref{theorem:extended_C-cost}.
(See Appendix~\ref{sec:ecost_theorems} for the proof.)

    \begin{corollary} \label{thm:equiangular}
    For the minimum number of parameters $(N=N_\mathrm{min})$, the C-cost $C(A)$ in Eq.~\eqref{eq:ecost} is always $C(A) \geq 1$ with equality if and only if the parameter configurations $\{\bm{q}_i\}_{i=1}^{N}$ satisfy
    \begin{equation}
        |\bm{q}_i\cdot\bm{q}_{j}| = \frac{1}{\sqrt{d+2}}
         ~ (\mathrm{for ~all} ~ i\neq j).
    \end{equation}
    \end{corollary}

The equality condition in Corollary~\ref{thm:equiangular} tells us that the parameters must be equiangular unit vectors.
This equiangular property is known as {\it equiangular lines} in real spaces~\cite{LemmensSeidel73,lemmens1991equiangular,greaves2016equiangular,jiang2021equiangular}, which is equivalent to the algebraic graph theory of \textit{two-graphs}~\cite{godsil01}.  The existence of $N_{\mathrm{min}} = d(d+1)/2$ equiangular lines in $\mathbb{R}^d$ is known as the \textit{Gerzon bounds}, and so far only shown to hold for $d = 2, 3, 7, 23$. For our optimal parameter configurations, only the cases of Rx and Fraxis gates ($d=2, 3$), there exists a {\it unique} set of $N_\mathrm{min}$ equiangular unit vectors (up to rotation) and such parameter configuration {\it uniquely} achieves the minimum value of C-cost $C(A)$. The non-existence of such optimal parameter configuration for FQS gate ($d = 4$) is due to the non-existence of equiangular lines satisfying the condition of Corollary~\ref{thm:equiangular}, which is attributed to Haantjes~\cite{haantjes1948equilateral} and  Neumann in~\cite{LemmensSeidel73} (see also~\cite{LinYu2020}).

\subsection{The Rotation Invariance of C-cost.}~\label{subsec:RotInvSec}
The C-cost $C(A)$ in Eq.~\eqref{eq:ecost} is invariant to rotation of all the parameter configurations. 
In other words, a parameter configuration $({\bm q}_1,..., {\bm q}_K)$ and its rotated configuration $(R{\bm q}_1,..., R{\bm q}_K)$ have the same value of the C-cost, where $R$ is a rotation matrix $ \in {\mathbb R}^{d\times d} (R^TR=I)$. 
See Appendix~\ref{apdx:rotation_invariance} for the proof of rotation invariance.
This implies that, for any parameter $\bm{q}$ of a single-qubit gate of interest, there exists the optimal parameter configuration $\{\bm{q}_i\}$ such that $\bm{q}\in \{\bm{q}_i\}$.
This property allows for one reduction of the total number of measurements required in the matrix construction, i.e. reduced to two for Rotosolve, five for Fraxis, and nine for FQS by diverting the previous results to the subsequent gate update. The reduction for Rotosolve has been known before~\cite{ostaszewski2021} but not for Fraxis and FQS. In each step of the sequential optimizations, the resulting cost value after the parameter update can be estimated without additional measurement.
Since all parameters are fixed except for that of the target gate, this estimated cost can be regarded as one of the observable expectation value $b_1$ in the subsequent application, where the parameter $\bm{q}_1$ of the next gate of interest is diverted from the previous application.

\begin{figure}[t]
\begin{algorithm}[H]
\caption{Algorithm to reuse optimization results of the previous gate}\label{algo:resusing} 
\begin{algorithmic}[1]
\Require The parameter $\bm{q}^\mathrm{pre}$ of the target gate,
The estimated minimum eigenvalue $\lambda^\mathrm{pre}$ in the previous FQS, and the optimal parameter configurations $\{\bm{q}^*_1, \cdots, \bm{q}^*_N\}$.
\Ensure the optimized parameter of the target gate $\bm{q}^\mathrm{opt}$ and the updated cost $\lambda$. 
\State Find a rotation matrix $R$ such that $\bm{q}^\mathrm{pre}=R \bm{q}_1^*$
\State Set $b_1 = \lambda^\mathrm{pre}$ (instead of measuring $b_1 = \braket{H}(\bm{q}^*_1$))
\For {$i = 2$ to $N$}
        \State Measure $b_i = \braket{H}(R\bm{q}^*_i)$
\EndFor
\State Set $G = \mathrm{vech}^{-1}(DA^+\bm{b})$.
\State Diagonalize $G$ and obtain the minimum eigenvalue $\lambda$ and the corresponding eigenvector $\bm{q}^\mathrm{opt}$
\State Return $\bm{q}^\mathrm{opt}$ and $\lambda$  
\end{algorithmic}
\end{algorithm}
\end{figure}

The detailed procedure is as follows; (1) Prepare an optimal parameter configuration $\{\bm{q}^*\}$, the gate parameter set $\{\bm{q}^{(m)}\}$ for $m = 1,\cdots,M $, and the temporal cost value $\braket{H}(\{\bm{q}^{(m)}\})$ where $m$ and $M$ denote the gate index and the total number of parametrized gates, respectively.
(2) Finds a rotation matrix $R$ such that $\bm{q}_1^*=R^T\bm{q}^{(m)}$ where the $m$th gate is of interest and sets $b_1=\braket{H}$.
(3) Measure the cost values with the parameter  $\{R \bm{q}_i^*\}$ for $i=2,\cdots N_\mathrm{min}$ setting $b_i=\braket{H}(R\bm{q}_i^*)$.
(4) Construct the matrix $G$ from $\bm{b}$ and $\{R \bm{q}_i^*\}$ 
(5) Diagonalise the matrix to estimate the new parameter $\bm{q}^{(m)}$ and the new cost $\braket{H}$, which can be reused in the next iteration and go back to (2) until convergence.
The pseudo-code of this procedure is given in Algorithm~\ref{algo:resusing}.

\subsection{Optimal configurations} \label{subsec:optMCs}
The minimum size of parameter configuration ($N_\mathrm{min}$) for Rx, Fraxis, and FQS are 3, 6, and 10, respectively. 
According to Corollary~\ref{thm:equiangular} 
in the case of the Rx gate, 
the three equiangular vectors on a unit circle are trivially represented by $\bm{q} = [\cos{\frac{2}{3}\pi n\theta }, \sin{\frac{2}{3}\pi n\theta}]^T$ for $n=0,\pm 1$, that is, the vector angle $\Delta\theta=2\pi/3$ (equivalently $\pi/3$) 
as shown in Figure~\ref{fig:optimalMCs}A.
In contrast, the original parameter configuration proposed in Rotosolve~\cite{ostaszewski2021} was $\Delta\theta=\pi/2$, which resulted in $C(A)=3/2$.
(It is worth noting that in~\cite{nakanishi2020} it is argued that arbitrary parameter configurations can be used due to the sine property of the expectation value but did not discuss the estimation accuracy dependent on the parameter configurations under the finite measurements.)
To achieve the same estimation accuracy, our optimal parameter configuration ($\Delta\theta= 2\pi/3$) requires two-thirds as many shots as the original parameter configuration ($\Delta\theta=\pi/2$). 

Corollary~\ref{thm:equiangular} is also instrumental for Fraxis with $d=3$. It is also possible to find the equiangular formation of six unit vectors in 3D space.
Figure~\ref{fig:optimalMCs}B shows the unique optimal parameter configuration except for the rotational degrees of freedom, where they form a regular icosahedron.
The original parameter configuration of Fraxis has $C(A) = 1.8$ \cite{watanabe2021} (See Appendix \ref{apdx:used_parameter}). 
Thus, the optimal configuration improves the estimation accuracy $1.8$ times with the consistent number of shots.

In contrast, it was proved that $N_\mathrm{min}$ ($=10$) equiangular unit vectors cannot be placed in $d$ ($=4$) dimensional Euclidean space.
Namely, Corollary~\ref{thm:equiangular} tells that there is no parameter configuration that satisfies $C(A)=1$ for $N=10$.
In addition, Corollary~\ref{thm:equiangular} also implies that the minimum size of the parameter configuration ($N=10$) may not be the most efficient if the total number of shots are limited for a single FQS execution, although it is not straightforward to know the analytical minimum value and the corresponding parameter configurations.
Instead, we searched the numerical solution by classical optimization, where $C(A)$ is minimized based on the gradient descent method. 
Since the algorithm may lead to a local minimum solution, we repeated the algorithm independently $10^5$ times starting from random initial configurations. 

For $N=10$, we have obtained the same optimized C-cost value ($C(A)\approx 1.033172$) from all the initial configurations as far as our experimental trials, which implies that all simulations presumably reached to the global minimum.
Although the obtained configurations were not numerically identical,
we found that they were attributed to a unique configuration just by reversal and rotational operations.
Since the reversal of each parameter does not affect the expectation value (i.e., $h(\bm{q})=h(-\bm{q})$) and the uniform rotation of the parameter configuration gives the indentical value of the C-cost (See Sec.~\ref{subsec:RotInvSec}), 
all the configurations were equivalent, which seem to be optimal.

Figure~\ref{fig:optimalMCs}C shows the unique optimal parameter configurations for the FQS case.
In this figure, the parameter configurations are projected into 3-dimensional space by a stereographic projection. It means that extra 1D components that cannot be displayed are projected in the radial direction. 
See Appendix~\ref{apdx:used_parameter} for the parameter values of the optimal and other parameter configurations. 
From the parameter values of the (numerically obtained) optimal parameter configurations (Eq.~\eqref{eq:opt_MC_FQS}), we can see the optimal parameter configurations has highly symmetrical structure; the first four parameters $\{\bm{q}_1, ..., \bm{q}_4\}$ and its opposite $\{-\bm{q}_1, ..., -\bm{q}_4\}$ constitute a regular cube in a hyperplane and the last six parameters $\{\bm{q}_5, ..., \bm{q}_{10}\}$ constitute a regular octahedron in a hyperplane (its opposite also constitute another regular octahedron) as shown in Fig.~\ref{fig:optimalMCs}.

For FQS, the original parameter configuration  has $C(A) = 3.0$ and the optimal parameter configurations estimated with numerical experiments is approximately $C(A)\approx 1.033172$.
And thus, to achieve a certain accuracy, the optimal parameter configuration reduces the number of required shots 3 times than that of the original.

\begin{table}[t]
\caption{\textbf{C-cost values for the different sizes of parameter configurations of FQS.} (A) Comparison of the C-cost $C(A)$ with a constant number of shots for evaluating an expectation value. (B) Comparison of scaled C-cost $(N-1)C(A)/N$ with constant number of shots per single-gate optimization.}
\centering
\label{tbl:Redundant}
\begin{tabular}{c|ccc}
N & 10 & 11 & 12 \\ \hline
(A) & 1.03317 & 1.00539 & \textbf{1.00000}  \\
(B) & 0.92985 & \textbf{0.91399} & 0.91667 \\
\hline
\end{tabular}
\end{table}

Likewise, we also conducted the numerical optimization to find the optimal parameter configuration for redundant measurements with $N=11, 12$.
As a result, all the optimizations converged to a consistent value of $C(A)$ within computational precision, which is consistent with the case of $N=N_\mathrm{min}$. 
However, the optimal configurations are not necessarily unique, which is in contrast to $N=N_\mathrm{min}$.
While the obtained $C(A)$ was $\approx$ 1.005390 for $N=11$,  $C(A)$ was exactly converged to unity for $N=12$.
It is also notable that the optimal configurations of $C(A)=1$ for $N=12$ include the regular 24-cell polytope in 4-dimensional space as shown in Fig.~\ref{fig:optimalMCs}D.

Therefore, If the total number of shots for $A$ matrix construction is constant, the optimal sizes of $N$ are three for Rotosolve, six for Fraxis, and twelve for FQS.

Next, we focus on the optimal $N$ allowing the reduction of measurements exploiting the rotation invariance as mentioned in Sec.~\ref{subsec:RotInvSec}.
Assuming a constant number of shots per gate, the measurement reduction modifies the relation between $C(A)$ and $s_\mathrm{min}$ as 

\begin{equation}\label{eq:error_reused}
    \mathrm{Var}[\lambda_1^*] = \frac{\sigma^2}{s_\mathrm{min}}\frac{N-1}{N}C(A). 
\end{equation}
where the C-cost is apparently scaled by $(N-1)/N$.
Note that this factor does not change the optimal $N$ for Rotosolve and Fraxis.
Thus, 
it is most efficient to revert the estimated value in previous optimization to construct $A$ and additionally execute two and five measurements for Rotosolve and Fraxis, respectively. 
It is worth noting that Table.~\ref{tbl:Redundant} shows that
the optimal $N$ for FQS is shifted from twelve to eleven by measurement reduction, although the difference is smaller than 1 \%.  
Altogether, under limitation of the total number of shots, it is most efficient to construct the matrix $A$ by three-, six-, and twelve-type measurements for the expectation values
in the beginning of Rotosolve, Fraxis, and FQS optimizations, respectively.
In contrast, during the sequential optimization, matrix $A$ should be made by one estimation value from the previous step and two, five, and ten values from subsequent measurements of Rotosolve, Fraxis, and FQS, respectively.  

It should be also noted that this optimal condition may differ depending on the supposed condition of real devices.
For instance, if parallel computation is allowed, where a constant number of shots are available for evaluating an expectation value even though when $N$ varies, $C(A)$ would not be an appropriate metric because the assumption about the number of shot is not valid, and thus one should trivially employ as large $N$ as as possible.

\begin{figure*}[t]
    \includegraphics[width=0.8\linewidth]{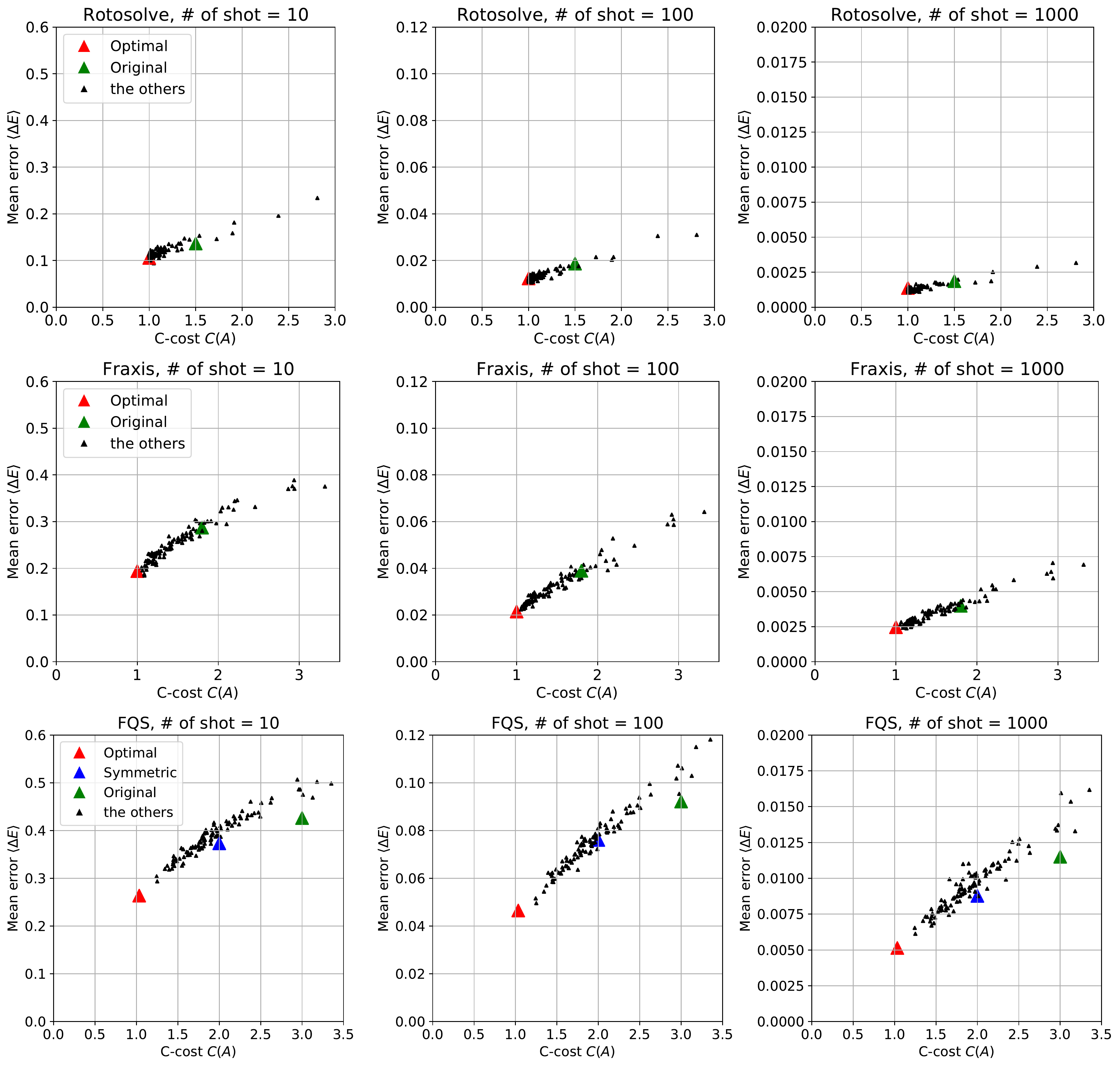}
    \caption{ \textbf{The average energy error for one-time optimization with different parameter configurations.
    }
Each subplot shows the averaged deviations of the estimated minimum from the true minimum energy (vertical axis), where the former energy was evaluated from $G$ made with randomly generated parameter configurations under the limited total number of shots, while the latter energy were obtained {\it statevector simulator}.
The left, center, and right columns show the results with 10, 100, and 1000 shots per circuit, respectively. 
The top, center, and bottom rows show the results for Rotosolve, Fraxis, and FQS, respectively. 
The results of original and optimal parameter configurations are highlighted in the figure.
The description about the number of shots above each subplot represents the number of shots used for a single mean value of the Hamiltonian based on a parameter configuration.
}
    \label{fig:ecost_of_difcof}
\end{figure*}

\section{Experiments}\label{sec:experiments}
In the following, we provide several experiments to numerically verify our proposed method on the condition of $N=N_\mathrm{min}$.

\subsection{Estimation Accuracy of One-time Optimization with Different Parameter Configurations}\label{subsec:expB}

We focused on the one-time optimization rather than an entire VQE processes.
To this end, we examined the averaged error of FQS between the exact minimum and the estimated minimum energies with limited number of shots for several parameter configurations.
We used the 2-qubit Hydrogen molecule-like Hamiltonian \cite{bravyi2017tapering} defined as
\begin{equation}
    \label{eq:ham_exp1}
    H = I\otimes Z + Z\otimes I + X\otimes X
\end{equation}
in this experiment. 
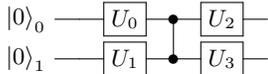
\begin{figure}
\[
\Qcircuit @C=1.0em @R=0.25em {
  \lstickx{\ket{0}_0} & \qw & \gate{U_{0}} & \ctrl{1} & \gate{U_{2}} & \qw \\
  \lstickx{\ket{0}_1} & \qw & \gate{U_{1}} & \ctrl{0} & \gate{U_{3}} & \qw &
}
\]
\caption{\textbf{2-qubit ansatz}}
\label{fig:2q_ansatz}
\end{figure}
We use the 2-qubit ansatz in Fig.~\ref{fig:2q_ansatz}, where we applied the corresponding single-qubit gate representation of Rotosolve (=RzRy), Fraxis, and FQS methods to $U_i$.
Here, the target gate to be optimized is $U_2$ for the FQS and Fraxis and the Ry gate in $U_2$ for the Rotosolve case. 
The experiments were performed as the following procedure. 
(1) prepare 100 independent parameter configurations, where the parameters of all the gates were randomly initialized with uniform probability distribution, which was followed by 50 iterations of the steepest decent optimization using $C(A)$ as a cost function.
(2) evaluate $A^+$ and $C(A)$ of the 100 parameter configurations.
(3) randomly initialize the PQC in the state-random manner for the respective single-qubit gates~\cite{watanabe2021}.
(4) obtain $\bm{b}$ (and $\bm{b}^*$) by the observable measurements based on the 100 parameter configurations,
and evaluate FQS matrices $G$ (and $G^*$) using the respective sets of $A^+$, $\bm{b}$ (and $\bm{b}^*$).
(5) execute FQS (Fraxis/Rotosolve) for $G$ (and $G^*$) to obtain $\bm{p}$ (and $\bm{p}^*$).
We repeat the procedure (3)--(5) $10^4$ times and evaluate the averaged error $\langle \Delta E \rangle$ in Eq.~\eqref{eq:EnergyDeviation} for each parameter configuration.
Note that we optimized the parameter configuration in process (1) above because the raw values of $C(A)$ distributed beyond $10^4$ otherwise.
In Fig.~\ref{fig:ecost_of_difcof}, we plotted 100 independent parameter configurations in $C(A)$ vs. $\braket{\Delta E}$ graph.
By definition, $C(A)$ and $\braket{\Delta E}$  are metrics to qualify the estimated energy and the estimated parameter, respectively.  
Although both the metrics are linked through the following equation, 
\begin{align}
\frac{N_\mathrm{min} d}{N}C(A) + \frac{sd(d-1) }{k\sigma^2} \mathbb{E}[\Delta E]  = \mathrm{Tr}[(A^TA)^{-1}],
\end{align}
the concrete behaviors are not necessarily trivial because of dependency on $A$ and the observable.
Here, we confirmed that the energy errors $\braket{\Delta E}$ are roughly proportional to $C(A)$ for all the cases, and $\braket{\Delta E}$ is inversely proportional to the number of shots approximately.
We also found that the optimal parameter configuration (red) achieves the lowest error$\braket{\Delta E}$, indicating that the optimal parameter configurations are actually effective to minimize the estimation error.
Although the magnitude of $\braket{\Delta E}$ in FQS is seemingly larger than that of Rotosolve, we note that it does not necessarily indicate the advantage of Rotosolve with respect to error suppression because the single gate expressibility is not comparable among the respective methods.
For instance, sequential Rotosolve applications of a series of three single-qubit gates are comparable to one-time FQS application.
In this case, however, it is not straightforward to compare them because of error propagation, which is beyond the present framework.
In the next section, instead, we examine the effect of the parameter configuration on the entire performance in comparison with the optimization methods.

\begin{figure}[t]
\leavevmode 
\Qcircuit @C=.2em @R=0.25em {
  \lstickx{\ket{0}_0} & \qw & \gate{U_{0}} & \ctrl{1}     & \qw  & \qw  & \qw & \qw  & \qw &\qw
      & \qw &\control \qw  
   &  \gate{U_{5l+4}} &\qw & \qw \\
  \lstickx{\ket{0}_1} & \qw & \gate{U_{1}} &  
      \control \qw & \gate{U_{5l}} &\ctrl{1}& \qw 
       & \qw& \qw&  \qw&  \qw & \qw &\qw &\gate{U_{5L+5}}  &\qw  \\
  \lstickx{\ket{0}_2} & \qw & \gate{U_{2}} & \qw   
      & \qw & \control \qw & \gate{U_{5l+1}}
      & \ctrl{1}&\qw& \qw&  \qw & \qw&\qw&\gate{U_{5L+6}}  
       & \qw  \\
  \lstickx{\ket{0}_3} & \qw & \gate{U_{3}} & \qw 
      & \qw  & \qw \qw & \qw & \control \qw 
      & \gate{U_{5l+2}} & \ctrl{1} \qw &\qw & \qw&\qw &\gate{U_{5L+7}} 
       & \qw   \\
  \lstickx{\ket{0}_4} & \qw & \gate{U_{4}} & \qw  
      & \qw & \qw & \qw   & \qw & \qw& \control \qw & \gate{U_{5l+3}} & \ctrl{-4}\qw&\qw& \gate{U_{5L+8}}  & \qw  \\
  &&&&&&&&&&&& \arrep{lllllllll}
  \gategroup{1}{4}{6}{13}{.25em}{--}
}
\caption{ \textbf{Cascading-block ansatz} }
\label{fig:cbansatz}
\end{figure}
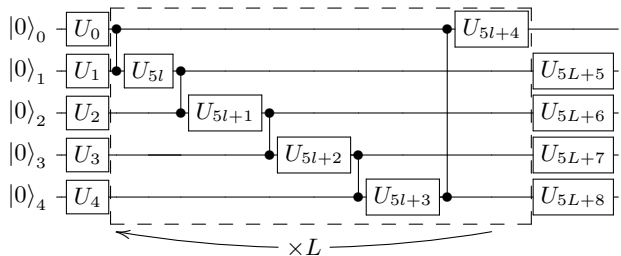

\begin{figure*}[t]
    \includegraphics[width=1.0\linewidth]{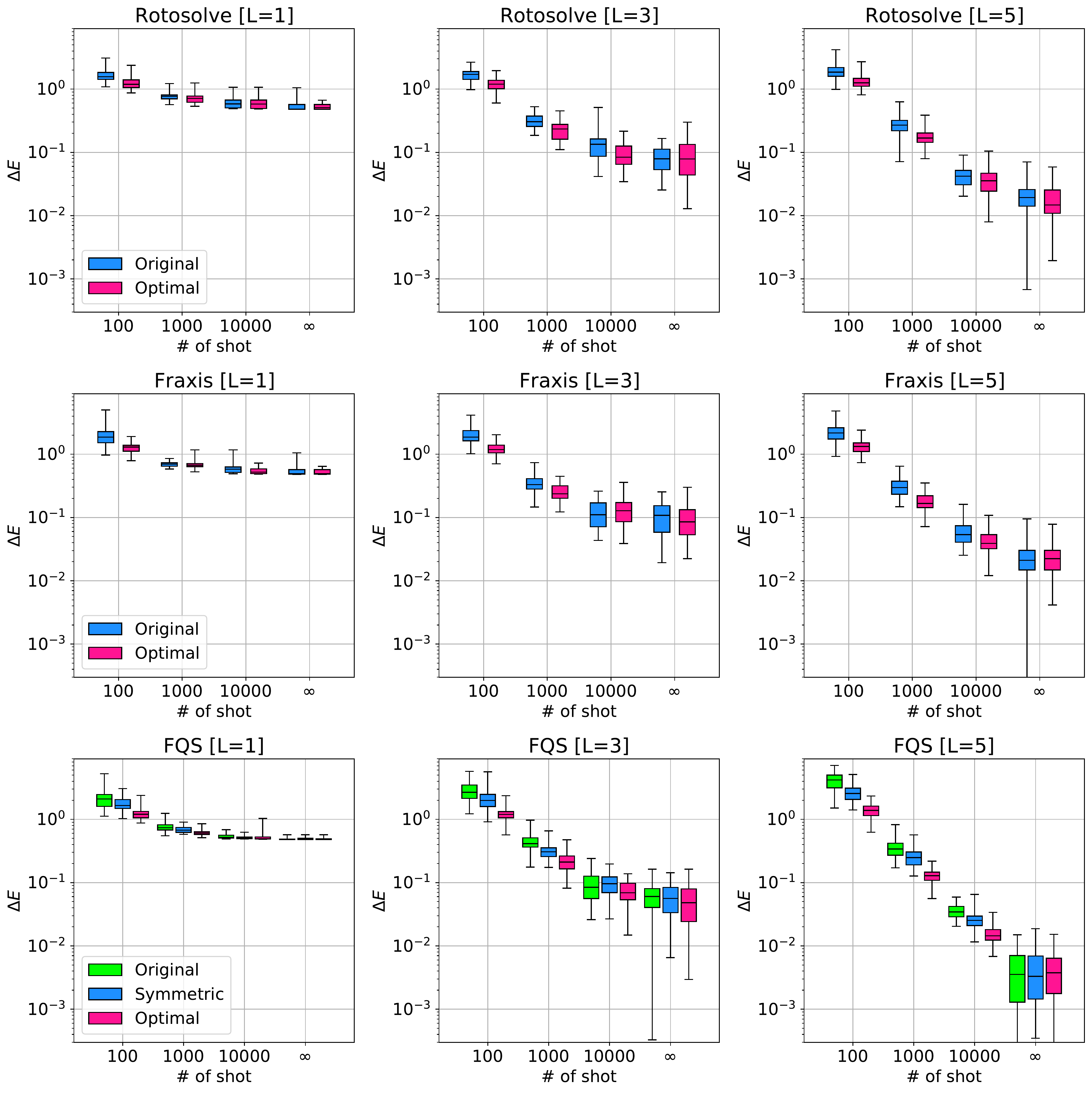}
    \caption{ \textbf{
    Comparison of the VQE performance with different parameter configurations.} 
    The vertical axis represents the deviation of the resulting VQE energy from the exact ground state energy. The respective energies were evaluated after 100 sweeps, where one sweep stands for sequential updates of all the single-qubit once. 
    The box-and-whisker plots show the statistics of the energy deviations ($\Delta E$) obtained by independent 100 VQE runs. 
    We carried out the VQEs with the circuit layers ($L$) from 1, 3, 5 and showed the results in respective subplots.
       }
    \label{fig:RC1}
\end{figure*}

\subsection{The Influence of the C-cost on VQE Performances}
\label{subsec:expC}

We investigate the effect of different parameter configurations on the results of VQE when we sequentially optimize single-qubit gates in quantum circuits by the framework of FQS~\cite{wada2022full}.
We employed the 5-qubit quantum Heisenberg model~\cite{bonechi1992heisenberg} defined as
\begin{equation}
    H = J\sum_{i=1}^5{ \sum_{\sigma=X,Y,Z}{ \sigma_i \sigma_{i+1} } }
    + h\sum_{i=1}^5{ Z_i }
\end{equation}
where $\sigma_i=I^{\otimes i-1}\otimes \sigma \otimes I^{\otimes 5-i}$($1\leq i\leq 5$), $\sigma_6=\sigma_1$.
We herein set $J=h=1$.
We used a Cascading-block ansatz shown in Fig.~\ref{fig:cbansatz}, where the gates within the dashed lines are repeated $L$ times.
We set $L=1, 3, 5$ in this experiment.
According to the optimization method, we applied the respective single-qubit representations to $U_i $ in the PQC.
We begin VQE with randomly initialize PQC in the state-random manner for respective single-qubit gates in the PQC.
In VQE, we sequentially applied Rotosolve/Fraxis/FQS to $U_i$ in the order of subscripts in Fig.~\ref{fig:cbansatz}, i.e., from the top-left to the right bottom. 
We term this procedure to update all gates in the PQC once as \textit{sweep}. 
In a single VQE run, we carried out 100 sweeps to obtained the estimated minimum eigenvalue $E$ of the Hamiltonian.
We performed independent 100 VQE runs and plotted the error distribution $\Delta E := E - E^*$ for respective 100 trials in Fig.~\ref{fig:RC1}, where $E^*$ is the exact minimum eigenvalue of the Hamiltonian.
We evaluated the resulting distributions using the number of shots to 100, 1000, 10000, and $\infty$ for two or three different parameter configurations (See Appendix~\ref{apdx:used_parameter} for the specific parameter values). 
Note that we used a statevector for VQE with an infinite number of shots.
Figure~\ref{fig:RC1} suggests that parameter configurations strongly affect the entire VQE performance and shows that the optimized parameter configuration ($C(A)\simeq1$) achieves the smallest errors on all the conditions with the finite numbers of shots. 
The optimal parameter configuration works more effectively 
as the number of shots is smaller, which is in line with the analysis of the one-time application to a single-qubit gate in Fig.~\ref{fig:ecost_of_difcof}. 
In addition, the impact of the parameter configuration on the VQE performance is not visible on shallow circuit and more distinct as the number of the layer increases.
In general, more expressive ansatz can potentially approximate the state of interest with higher precision.
Correspondingly, one has to increase the number of shots, because for accuracy $\epsilon$, the number of required shots scales in $\mathcal{O}(1/\epsilon^2)$. 
Otherwise, the enhanced expressibility by the circuit extension may not be highlighted.
Since the gain of $C(A)$ is equivalent to the increase of measurements, the optimal parameter configuration will be more beneficial as desired accuracy in VQE becomes higher.
In fact, FQS is superior to Rotosolve and Fraxis and the statevector simulation implies that FQS with ansatz of $L=5$ can potentially achieve the accuracy $\epsilon < 10^{-2}$.
However, it is less likely to reach this energy level with the 10000 shots which is a practical standard for the present quantum devices, i.e. IBM-Q device.
There, the parameter optimization assists VQE lowering the reachable energy level distinctively, although it is not the case for Rotosolve and Fraxis because the number of shots available are sufficient relative to their expressibility.

\section{Conclusions}
In this work, we showed that the parameter configuration affects the performance of analytical optimization of a single-qubit gate. 
This estimation error was quantified by the C-cost $C(A)$, the variance of the estimated value of the cost function.
We theoretically proved that the lower bound of $C(A)$ is unity.
We also showed that when the size of the parameter configuration is minimal, the C-cost becomes unity if and only if the parameter configuration satisfies the equiangular condition.
Exploiting this property, we found the optimal parameter configuration for Rotosolve and Fraxis.
Although we revealed no parameter configuration of minimum size for FQS achieves $C(A)=1$, it turned out the parameter configuration of $N=12$ corresponding to the regular 24-cell polytope in the 4-dimensional space satisfies $C(A)=1$.
In addition, we also demonstrated how to reduce the number of measurements for matrix construction by making use of the rotation invariance of $C(A)$. 
Then, the optimal parameter configurations exhibited the best results improving efficiency 1.5 times for Rotosolve, 1.8 times for Fraxis, and 3.0 times for FQS, when compared to the original parameters.
Additional numerical experiments showed that the parameter configuration affects the performance of not only the one-time optimization but also the entire VQE. 
We also found that the parameter configuration is more instrumental to elicit the VQE performance as the ansatz becomes more expressive.

\section*{Acknowledgements}
R.R. would like to thank Prof. David Avis of Kyoto University for the discussion on equiangular lines. 
H.C.W. was supported by JSPS Grant Numbers 20K03885. 
N.Y. and H.C.W were supported by the MEXT Quantum Leap Flagship Program Grant Number JPMXS0118067285 and JPMXS0120319794.

\input{appendix}

\bibliography{references}

\end{document}

%% file: appendix.tex
\appendix

\setcounter{theomain}{0}
\setcounter{corollary}{0}

\section*{Appendix}
\subsection{Free Quaternion Selection}
\label{apdx:FQS}

We show the minimum value of Eq.~\eqref{eq:hmat} is the minimum eigenvalue  $\lambda_1$ of $G$ achieved when $\bm{q} = \bm{p}_1$ for the corresponding eigenvector $\bm{p}_1$ of $G$. 
%

For the Lagrange multiplier method, we first define a function, $l(\bm{q},\lambda)$, corresponding the above optimization problem as
\begin{equation}
    l(\bm{q},\lambda) = \bm{q}^T G \bm{q} - \lambda ( \|\bm{q}\|^2 -1 ),
\end{equation}
where $\lambda$ is a Lagrange multiplier.
Taking the partial derivatives for $l(\bm{q},\lambda)$ and setting them to zero, we can obtain
\begin{eqnarray}
    G\bm{q} = \lambda \bm{q}.
\end{eqnarray}
Thus, the candidates for the local minimum/maximum value of $l(\bm{q},\lambda)$ and the solutions are the eigenvalues $\lambda_i$ and its normalized eigenvectors $\bm{p}_i$, respectively.

Substituting the normalized eigenvectors $\bm{p}_i$ into Eq.~\eqref{eq:hmat}, we get
\begin{align}
    \langle H \rangle 
    =\bm{p}^T_i G \bm{p}_i
    =\bm{p}^T_i (\lambda_i \bm{p}_i)
    =\lambda_i,
\end{align}
this means the global minimum value of Eq.~\eqref{eq:hmat} and its solution are given by the minimum eigenvalue $\lambda_1$ and the corresponding normalized eigenvector $\bm{p}_1$.

\subsection {Derivation of Analytical form of the Measures}
\subsubsection{Expectation value over an orthogonal basis}
\label{apdx:OrthogoIntegral}

We show several equations that are useful for derivation of analytical form of the measures.

Let $Z \in \mathbb R^{d\times d}$ be a random symmetric matrix which 
satisfies $\mathbb{E}[Z_{ij}] = 0$ for all $i,j$. 
Independently, let $P=(\bm{p}_1,...,\bm{p}_d)^T \in \mathbb{R}^{d\times d}$ be a random orthogonal matrix (i.e., the matrix is uniformly sampled from the orthogonal group $O(4)$).
Then, the following equations holds:
\begin{align}
\mathbb{E}[\bm{p}_i^{*T} Z \bm{p}_{j}^*] 
    =&\sum_{k,l} \mathbb{E}[ (\bm{p}_i^*)_k (\bm{p}^*_{j})_l Z_{kl} ]\notag \\
    =&\sum_{k,l} \mathbb{E}[ (\bm{p}_i^*)_k (\bm{p}^*_{j})_l ] \mathbb{E}[ Z_{kl} ] \notag\\
    =&~0,
\end{align}
and so,
\begin{align}\label{apdx:pzp_variance}
    \mathrm{Var}[\bm{p}_i^{*T} Z \bm{p}_{j}^*]
    =& \mathbb{E}[(\bm{p}_i^{*T} Z \bm{p}^*_{j})^2]
    - \mathbb{E}[\bm{p}_i^{*T} Z \bm{p}^*_{j}]^2 \notag \\
    =& \mathbb{E}[(\bm{p}_i^{*T} Z \bm{p}^*_{j})^2].
\end{align}

For $i=j$,

\begin{align}\label{apdx:pZp_variance2}
    \mathrm{Var}[\bm{p}_i^{*T} Z \bm{p}_{i}^*]
    =& \mathbb{E}[(\sum_{k,l} (\bm{p}_i)_k Z_{kl} (\bm{p}_j)_l)^2] \notag\\
    =& \sum_{k,l,m,n} \mathbb{E}[ (\bm{p}_i)_k (\bm{p}_i)_l (\bm{p}_i)_m (\bm{p}_i)_n ]\mathbb{E}[ Z_{kl} Z_{mn} ]\notag \\
    =& \sum_{k(=l)} \sum_{m(=n\neq k)} \mathbb{E}[ (\bm{p}_i)_k^2 (\bm{p}_i)_m^2] 
    \mathbb{E}[ Z_{kk} Z_{mm} ] \notag\\
    &+ \sum_{k(=m)} \sum_{l(=n\neq k)} \mathbb{E}[ (\bm{p}_i)_k^2 (\bm{p}_i)_l^2] 
    \mathbb{E}[ Z_{kl}^2 ] \notag \\
    &+ \sum_{k(=n)} \sum_{l(=m\neq k)} \mathbb{E}[ (\bm{p}_i)_k^2 (\bm{p}_i)_l^2] 
    \mathbb{E}[ Z_{kl} Z_{lk} ] \notag \\
    &+ \sum_{k}  \mathbb{E}[ (\bm{p}_i)_k^4] 
    \mathbb{E}[ Z_{kk}^2] \notag\\
    =&\sum_{k,l}\frac{\mathbb{E}[ Z_{kk} Z_{ll}]+ \mathbb{E}[ Z_{kl}^2] + \mathbb{E}[ Z_{kl} Z_{lk} ]}{d(d+2)}\notag\\
    =&\sum_{k,l} \frac{ \mathbb{E}[ Z_{kk} Z_{ll} ]+ 2\mathbb{E}[ Z_{kl}^2]}{d(d+2)}.
\end{align}
For the fourth equality, we employed the following relation. 
\begin{align}
    \mathbb{E}[(\bm{x}^T Z \bm{x})^2]
    =& d(d+2) \mathbb{E} [ ( \frac{\bm{x}^T}{\|\bm{x}\|} Z \frac{\bm{x}}{\|\bm{x}\|} )^2 ],
\end{align}
where $\bm{x}$ is a random vector in $\mathbb R^{d}$, which follows the $d$-dimensional multivariate standard normal distributions $\mathcal{N}(\bm{0},I)$, and $\bm{x}, Z$ are independent each other.

To evaluate Eq.~\eqref{apdx:pzp_variance} for $i\ne j$, we suppose another random vector $\bm{y}$ as $\bm{x}$, but independent of $\bm{x}$ and $Z$.
\begin{align} \label{eq:xzy2_mean}
    \mathbb{E}[\left(\bm{x}^T Z \bm{y}\right)^2]
    =& \mathbb{E}[(\sum_{i,j}(\bm{x})_i Z_{ij} (\bm{y})_j)^2] \notag\\
    =& \mathbb{E}[\sum_{i,j,s,t} (\bm{x})_i (\bm{x})_s Z_{ij} Z_{st} (\bm{y})_j (\bm{y})_t]\notag \\   
    =& \sum_{i,j,s,t} \mathbb{E}[ (\bm{x})_i (\bm{x})_s ]\mathbb{E}[ (\bm{y})_j (\bm{y})_t ]\mathbb{E}[ Z_{ij} Z_{st} ] \notag\\
    =& \sum_{i,j,s,t}\delta_{j,s} \delta_{i,t} \mathbb{E}[ Z_{ij} Z_{st} ]\notag\\    
    =& \sum_{i,j} \mathbb{E}[ Z_{ij}^2 ].
\end{align}
Here, we introduce two vectors as
\begin{equation}
    \bm{y}_{\parallel} = \frac{ (\bm{y}\cdot \bm{x}) }{ \|\bm{x}\|^2 }\bm{x}, ~~\bm{y}_{\perp} = \bm{y} - \bm{y}_{\parallel}.
\end{equation}
Using these vectors, we obtained the following relation,
\begin{align}\label{apdx:xZy_mean_v2}
    \mathbb{E}[(\bm{x}^T Z \bm{y})^2] 
    =&\mathbb{E}[(\bm{x}^T Z (\bm{y}_{\parallel}+\bm{y}_{\perp}))^2]\notag\\
    =&\mathbb{E}[(\bm{x}^T Z \bm{y}_{\parallel})^2] +
    \mathbb{E}[(\bm{x}^T Z \bm{y}_{\perp})^2]\notag\\
    &+
    2\mathbb{E}[(\bm{x}^T Z \bm{y}_{\parallel})(\bm{x}^T Z \bm{y}_{\perp})]\notag \\
    =&\mathbb{E}[(\bm{x}^T Z \bm{y}_{\parallel})^2] +
    \mathbb{E}[(\bm{x}^T Z \bm{y}_{\perp})^2].
\end{align}
For the third equality, we use the probability distribution $f$ satisfies
$f(\bm{x},\bm{y}_{\parallel},\bm{y}_{\perp})=f(\bm{x},\bm{y}_{\parallel},-\bm{y}_{\perp})$,
and thus
\begin{align}
    \mathbb{E}[(\bm{x}^T Z \bm{y}_{\parallel})(\bm{x}^T Z \bm{y}_{\perp})] =
    & \mathbb{E}[(\bm{x}^T Z \bm{y}_{\parallel})(\bm{x}^T Z (-\bm{y}_{\perp}))],
\end{align}
equivalently,
\begin{align}
     \mathbb{E}[(\bm{x}^T Z \bm{y}_{\parallel})(\bm{x}^T Z \bm{y}_{\perp})] 
    =& 0.
\end{align}
On the other hand, 
\begin{align} \label{apdx:pZpNot2_mean}
\mathbb{E}[(\bm{p}^{*T}_i Z \bm{p}^*_{j(\neq i)})^2]
    =&\mathbb{E}\left[\left(
    \frac{\bm{x}^T}{\|\bm{x}\|} Z \frac{\bm{y}_{\perp}}{\|\bm{y}_{\perp}\|}
    \right)^2\right] \notag \\
   =&\frac{1}{\mathbb{E}[\|\bm{x}\|^2]
    \mathbb{E}[\|\bm{y}_{\perp}\|^2]}
    \mathbb{E}[(\bm{x}^T Z \bm{y}_{\perp})^2]\notag \\
   =&\frac{1}{d(d-1)}
   \mathbb{E}[(\bm{x}^T Z \bm{y}_{\perp})^2]
\end{align}
where we suppose that the probability distribution $f$ satisfies $f(\bm{x}/\|\bm{x}\|,\bm{y}_{\perp}/\|\bm{y}_{\perp}\| )=f(\bm{p}_i,\bm{p}_{j(\neq i)}), \textit{a.e.}$, $\mathbb{E}[\|\bm{x}\|^2]=d$ and $\mathbb{E}[\|\bm{y}_{\perp}\|^2]=d-1$.

In addition the first term in Eq.~\eqref{apdx:xZy_mean_v2} 
\begin{align}\label{apdx:xZypar2_mean}
    \mathbb{E}[(\bm{x}^T Z \bm{y}_{\parallel})^2] 
    =&\mathbb{E}\left[\|\bm{x}\|^2 \|\bm{y}_{\parallel}\|^2\left(
    \frac{\bm{x}^T}{\|\bm{x}\|} Z \frac{\bm{y}_{\parallel}}{\|\bm{y}_{\parallel}\|}
    \right)^2\right] \notag \\    =&\mathbb{E}[\|\bm{x}\|^2]\mathbb{E}[\|\bm{y}_{\parallel}\|^2] 
    \mathbb{E}\left[\left(
    \frac{\bm{x}^T}{\|\bm{x}\|} Z \frac{\bm{y}_{\parallel}}{\|\bm{y}_{\parallel}\|}
    \right)^2\right] \notag \\
    =& d~\mathbb{E}[(\bm{p}^{*T}_i Z \bm{p}^*_{i})^2 ] \notag \\
    =&\sum_{k,l} \frac{ \mathbb{E}[ Z_{kk} Z_{ll} ]+ 2\mathbb{E}[ Z_{kl}^2]}{d+2}.
\end{align}
where the second equality arises from the independence of the random variables, and the third equality is based on 
$\bm{x}/\|\bm{x}\|=\bm{y}_{\parallel}/\|\bm{y}_{\parallel}\|$, $f(\bm{x}/\|\bm{x}\|) 
= f(\bm{p}_i), \textit{a.e.}$, $\mathbb{E}[\|\bm{x}\|^2]=d$ and $\mathbb{E}[\|\bm{y}_{\parallel}\|^2]=1$.

From Eq.~\eqref{apdx:pZp_variance2}\eqref{apdx:xZy_mean_v2}\eqref{apdx:pZpNot2_mean}\eqref{apdx:xZypar2_mean}, we finally obtain 
\begin{align}\label{apdx:pZpNot_variance}
    \mathrm{Var}[\bm{p}_i^{*T} Z \bm{p}^*_{j(\neq i)}]
    &=\mathrm{E}[(\bm{p}_i^{*T} Z \bm{p}^*_{j(\neq i)})^2]\notag\\
    &=\sum_{k,l} \frac{d \mathbb{E}[Z_{kl}^2]-\mathbb{E}[Z_{kk}Z_{ll}]}{d(d-1)(d+2)}. 
\end{align}

\subsubsection{Derivation of analytical form of \texorpdfstring{$\mathrm{Var}[\lambda_1(\bm{\epsilon})]$}{Var[λ(ε)]}}
\label{apdx:analytical_form}
 Using the noise model $\langle \epsilon_i \rangle = \bm{0}, \langle \epsilon_i \epsilon_j \rangle = \sigma^2 \delta_{i,j}/s$ and the uniformly distributed model of the first eigenvector $\bm{p}_1^*$. We show the analytical form of $\mathrm{Var}[\lambda_1(\bm{\epsilon})]$ in Eq.~\eqref{eq:score_var}: 
\begin{equation}
    \label{eq:x:secondmeasure}
    \mathrm{Var}[\lambda_1(\bm{\epsilon})] = \mathrm{Var}[\bm{p}_1^{*T} \mathrm{vech}^{-1}(D^{-1}A^+\bm{\epsilon}) \bm{p}_1^*] \notag.
\end{equation}
For simplicity, we write $Z=\mathrm{vech}^{-1}(D^{-1}A^+\bm{\epsilon})$.
Note that $Z$ is a symmetric matrix and satisfies $\mathbb{E}[Z]=O$ because $\mathbb{E}[D^{-1}A^+\bm{\epsilon}]=D^{-1}A^+\mathbb{E}[\bm{\epsilon}]=\bm{0}$. 
Thus, using Appendix~\ref{apdx:OrthogoIntegral}, 
\begin{equation}
    \label{eq:x:use_orthogo}
    \mathrm{Var}[\bm{p}_1^{*T} Z \bm{p}_1^*] = \\
    \frac{
    \sum_{i,j} \mathbb{E}[ Z_{ii} Z_{jj} ]
    + 2 \mathbb{E}[ Z_{ij}^2]
    }{d(d+2)}.
\end{equation}
Then, we deal with the first term $\sum_{i,j} \mathbb{E}[ Z_{ii} Z_{jj} ]$ and the second term $\sum_{i,j} \mathbb{E}[ Z_{ij}^2 ]$, separately.
To this end, we introduce some useful representations. 
We note 
Eq.~\eqref{eq:H_hq} can be rewritten as
\begin{eqnarray}
    \left<H\right> 
    &=& h(\bm{q})^T ~\bm{g} \notag\\
    &=& \left(\bm{q}^T\otimes\bm{q}^T\right) \mbox{vec}(G).
\end{eqnarray}
Here $\bm{q}$ is the parameter of the target single-qubit gate, $G$ is the FQS matrix, and $\mbox{vec}: {\mathbb R^{d\times d}} \rightarrow {\mathbb R
}^{d^2}$ is the vectorization operator for matrices.

Next, we introduce a linear transformation $L\in \mathbb R^{d^2 \times d(d+1)/2}$ between the vector $\bm{g}$ and $\mbox{vec}(G)$ in the Rx, the Fraxis, and the FQS gates as

\begin{equation}
    \normalsize
    \!\! L \!=\!\!
    {\scriptsize
    \begin{bmatrix}
    1 & 0 & 0 \\
    0 & 0 & c \\
    0 & 0 & c \\
    0 & 1 & 0
    \end{bmatrix} 
    }
    {\scriptsize ,~
    \begin{bmatrix}
    1 & 0 & 0 & 0 & 0 & 0 \\
    0 & 0 & 0 & c & 0 & 0 \\
    0 & 0 & 0 & 0 & c & 0 \\
    0 & 0 & 0 & c & 0 & 0 \\
    0 & 1 & 0 & 0 & 0 & 0 \\
    0 & 0 & 0 & 0 & 0 & c \\
    0 & 0 & 0 & 0 & c & 0 \\
    0 & 0 & 0 & 0 & 0 & c \\
    0 & 0 & 1 & 0 & 0 & 0
    \end{bmatrix}}
    {\scriptsize ,~
    \begin{bmatrix}
    1 & 0 & 0 & 0 & 0 & 0 & 0 & 0 & 0 & 0\\ 
    0 & 0 & 0 & 0 & c & 0 & 0 & 0 & 0 & 0\\ 
    0 & 0 & 0 & 0 & 0 & c & 0 & 0 & 0 & 0\\ 
    0 & 0 & 0 & 0 & 0 & 0 & c & 0 & 0 & 0\\ 
    0 & 0 & 0 & 0 & c & 0 & 0 & 0 & 0 & 0\\ 
    0 & 1 & 0 & 0 & 0 & 0 & 0 & 0 & 0 & 0\\ 
    0 & 0 & 0 & 0 & 0 & 0 & 0 & c & 0 & 0\\ 
    0 & 0 & 0 & 0 & 0 & 0 & 0 & 0 & c & 0\\ 
    0 & 0 & 0 & 0 & 0 & c & 0 & 0 & 0 & 0\\ 
    0 & 0 & 0 & 0 & 0 & 0 & 0 & c & 0 & 0\\ 
    0 & 0 & 1 & 0 & 0 & 0 & 0 & 0 & 0 & 0\\ 
    0 & 0 & 0 & 0 & 0 & 0 & 0 & 0 & 0 & c\\ 
    0 & 0 & 0 & 0 & 0 & 0 & c & 0 & 0 & 0\\ 
    0 & 0 & 0 & 0 & 0 & 0 & 0 & 0 & c & 0\\ 
    0 & 0 & 0 & 0 & 0 & 0 & 0 & 0 & 0 & c\\ 
    0 & 0 & 0 & 1 & 0 & 0 & 0 & 0 & 0 & 0   
    \end{bmatrix}},
\end{equation} respectively, where $c=1/\sqrt{2}$.
Note that the transformation satisfies
\begin{eqnarray}
    L^T L 
    &=& I, \\
    \mbox{vec}(G) 
    &=& L \bm{g}, 
\end{eqnarray}
And so,
\begin{eqnarray}
    \label{eq:x:hmql}
    (\bm{q}^T\otimes\bm{q}^T) L 
    &=& h(\bm{q})^T.
\end{eqnarray}
We may also consider the inverse transformation of vectorization $\mathrm{vec}^{-1}: {\mathbb R
}^{d^2} \rightarrow {\mathbb R^{d\times d}}$ as 
\begin{equation}\label{eq:vec-1}
    \mbox{vec}^{-1}(L\bm{g}) 
    = G~~~\forall \bm{g} \in {\mathbb R}^{d(d+1)/2},
\end{equation}
and $\mathrm{vech}^{-1}: {\mathbb R}^{d^2} \rightarrow {\mathbb R^{d\times d}}$,
\begin{equation}\label{eq:vech-1}
   \mathrm{vech}^{-1}(D^{-1}\bm{g})=G~~~\forall \bm{g} \in {\mathbb R}^{d(d+1)/2}.
\end{equation}
This leads to
\begin{equation}
    Z = \mathrm{vech}^{-1}(D^{-1}A^+\bm{\epsilon}) = \mathrm{vec}^{-1}(LA^+\bm{\epsilon}).
\end{equation}
Then, the first term in Eq.~\eqref{eq:x:use_orthogo} is rewritten as
\begin{eqnarray}\label{eq:ZiiZss}
    && \sum_{i,j} \mathbb{E}[ Z_{ii} Z_{jj} ] \notag \\
    &=& \mathbb{E}[~\sum_{i,j}
    \mathrm{vec}^{-1}(LA^+\bm{\epsilon})_{ii}
    \mathrm{vec}^{-1}(LA^+\bm{\epsilon})_{jj} ]\notag \\
    &=& \mathbb{E}[ (\mbox{vec}(I)^T LA^+\bm{\epsilon})(\mbox{vec}(I)^T LA^+\bm{\epsilon})]\notag\\
    &=& \frac{\sigma^2}{s} (\mbox{vec}(I)^T LA^+ \cdot\mbox{vec}(I)^T LA^+) \notag \\
    &=& \frac{\sigma^2 }{s}
    \mbox{vec}(I)^T LA^+ (A^+)^{T} L^T \mbox{vec}(I)\notag \\
    &=& \frac{\sigma^2}{s} \bm{1}_d^T A^+ (A^+)^{T} \bm{1}_d\\ 
    &=& \frac{\sigma^2}{s} \bm{1}_d^T (A^TA)^{-1} \bm{1}_d\\
\end{eqnarray}
where $I$ is the identity matrix, and $\bm{1}_d:=L^T \mbox{vec}(I)$ is the vector in $\mathbb{R}^{d(d+1)/2}$ whose first $d$ elements is unity and the rest are zero.
For the sixth equality, we used the folloing relation,
\begin{eqnarray}\label{eq:A+A+T}
    A^+(A^+){^T} 
    &=& ((A^TA)^{-1}A^T)((A^TA)^{-1}A^T)^T \notag \\
    &=& (A^TA)^{-1} A^TA (A^TA)^{-1}\notag\\
    &=& (A^TA)^{-1}.
\end{eqnarray}

The second term in Eq.~\eqref{eq:x:use_orthogo} is also rewritten as
\begin{eqnarray}\label{eq:ZijZij}
    &&\sum_{i,j} \mathbb{E}[Z_{ij}^2 ] \notag \\
    &=& \mathbb{E}[~\sum_{i,j} \mathrm{vec}^{-1}(LA^+\bm{\epsilon})_{ij}\mathrm{vec}^{-1}(LA^+\bm{\epsilon})_{ij} ]\notag \\
   &=&\mathbb{E}[ (LA^+\bm{\epsilon})^T(LA^+\bm{\epsilon})]\notag\\
   &=& \frac{\sigma^2}{s} \mathrm{Tr}[(LA^+)^T(LA^+) ] \notag\\
    &=& \frac{\sigma^2}{s} \mathrm{Tr} [A^+(A^+)^{T}L^TL]\notag\\
    &=& \frac{\sigma^2}{s} \mathrm{Tr} [(A^TA)^{-1}],
\end{eqnarray}

Summarizing Eqs.~\eqref{eq:x:use_orthogo},\eqref{eq:ZiiZss},and \eqref{eq:ZijZij},
$\mathrm{Var}[\lambda_1(\bm{\epsilon})]$ is expressed as

\begin{align}
    &\mathrm{Var}[\lambda_1]\notag\\
    &=\frac{\sigma^2}{s d(d+2)} \left(
    \bm{1}_d^T (A^TA)^{-1} \bm{1}_d +
    2\mathrm{Tr}[ (A^TA)^{-1} ]\
    \right) \notag\\
    &= \frac{\sigma^2}{sd(d+2)}
    \mathrm{Tr}[ (A^TA)^{-1}( \bm{1}_d \bm{1}_d^T + 2I)],
    \label{eq:x:ecost_final}
\end{align}
where the following identity:
\begin{equation}
    \bm{1}_d^T (A^TA)^{-1} \bm{1}_d = \mathrm{Tr}[(A^TA)^{-1} \bm{1}_d \bm{1}_d^T]
\end{equation}
is employed for the last equality.

\subsubsection{Discussion of \texorpdfstring{$\mathbb{E}[\Delta E]$}{} for the perturbation effect}
\label{apdx:alternative_measure}

Using the second-order perturbation theory of matrix \cite{cha2018perturbation}, the energy error $\Delta E$ is approximated as 
\begin{equation} \label{eq:DeltaE}
    \Delta E = \sum_{i>1}^d \frac{ (\bm{p}_i^{*T} \mathrm{vech}^{-1}(D^{-1}A^+\bm{\epsilon}) \bm{p}_1^* )^2 }{ \lambda_i^* - \lambda_1^* }.
\end{equation}
Note that Eq.~\eqref{eq:DeltaE} is not applicable when the lowest-energy eigenstate is degenerated. 
However, the following argument has been found to hold well experimentally. 
This equation leads to:
\begin{equation} \label{eq:mean_delta_E}
    \mathbb{E}\left[\Delta E \right] = \mathbb{E}\left[
    \sum_{i>1}^d \frac{ (\bm{p}_i^{*T} \mathrm{vech}^{-1}(D^{-1}A^+\bm{\epsilon}) \bm{p}_1^* )^2 }{ \lambda_i^* - \lambda_1^* }
    \right].
\end{equation}
However, unlike $\mathrm{Var}[\bm{p}_1(\bm{\epsilon})]$,
this measure also depends on the probability distribution $f(\lambda_1^*,...,\lambda_d^*)$ of the eigenvalues of the matrix $G$.
Assuming these eigenvalues are independent of each other, that is,
\begin{equation}
    f(\lambda_1^*,...,\lambda_M^*) = \prod_{i=1}^d f(\lambda_i^*),
\end{equation}
and the matrix of the eigenvectors $P=(\bm{p}_1,...,\bm{p}_d)^T$ is a random orthogonal,
Eq.~\eqref{eq:mean_delta_E} can be written as
\begin{align}
    \mathbb{E}\left[\Delta E \right]
    =& \sum_{i>1}^d
    \mathbb{E}\left[
    \frac{1}{ \lambda_i^* - \lambda_1^* }
    \right]
    \mathbb{E}\left[
    (\bm{p}_i^{*T} \mathrm{vech}^{-1}(D^{-1}A^+\bm{\epsilon}) \bm{p}_1^* )^2
    \right] \notag\\
    =& k\mathbb{E}\left[
    (\bm{p}_2^{*T} \mathrm{vech}^{-1}(D^{-1}A^+\bm{\epsilon}) \bm{p}_1^* )^2
    \right],
\end{align}
where $k := \sum_{i>1}^d \mathbb{E}[ (\lambda_i^* - \lambda_1^*)^{-1}]$.
This means the measure $\mathbb{E}\left[\Delta E \right]$ can be evaluated with some modeling of the true FQS matrix $G$ and the measurement errors $\bm{\epsilon}$.

For simplicity, we now write $Z=\mathrm{vech}^{-1}(D^{-1}A^+\bm{\epsilon})$.
From Eq.\eqref{apdx:pZpNot_variance} in Appendix~\ref{apdx:OrthogoIntegral}
, 
\begin{align}
    \mathbb{E}\left[\Delta E \right]
    &=k\sum_{i,j} \frac{d \mathbb{E}[Z_{ij}^2]-\mathbb{E}[Z_{ii}Z_{jj}]}{d(d-1)(d+2)}\\
    &= \frac{k\sigma^2}{sd(d-1)(d+2)} \notag \\
    &~\times\left(d~\mathrm{Tr}[(A^TA)^{-1}]  -\bm{1}_d^T (A^TA)^{-1}\bm{1}_d\right)\notag\\
    &= \frac{k\sigma^2}{sd(d-1)(d+2)} \notag \\
    &~\times \mathrm{Tr}[(A^TA)^{-1}(dI -\bm{1}_d\bm{1}_d^T)] \label{eq:e_de_analy}
\end{align}
In addition, if $C(A)=1$, i.e. the case of theoretical lower bound, $A^TA=\frac{N}{d(d+2)}(2I+\bm{1}_d\bm{1}_d^T)$ holds from Theorem~\ref{theorem:extended_C-cost}. 
As a result, we obtain
\begin{align}
    \mathbb{E}\left[\Delta E \right] =
    \frac{k\sigma^2}{s} \frac{ d(d+2) }{4N} \label{eq:e_de_analy_Ceq1},
\end{align}
where we used the following relation,
\begin{align}
    &\mathrm{Tr}\left[(A^TA)^{-1}(dI -\bm{1}_d\bm{1}_d^T)\right] \notag\\
    &= \frac{d(d+2)}{N}\mathrm{Tr}\left[(2I+\bm{1}_d\bm{1}_d^T)^{-1}(dI -\bm{1}_d\bm{1}_d^T)\right] \notag \\
    &= \frac{d(d+2)}{N}\mathrm{Tr}\left[\left(\frac{1}{2}I-\frac{1}{2(d+2)}\bm{1}_d\bm{1}_d^T\right)(dI -\bm{1}_d\bm{1}_d^T)\right] \notag\\ 
    &= \frac{d(d+2)}{N}\mathrm{Tr}\left[\frac{d}{2}I-\frac{1}{2}\bm{1}_d\bm{1}_d^T\right] \notag\\ 
    &= \frac{d(d+2)}{N} \frac{d}{2} (N_\mathrm{min}-1) \notag\\ 
    &= \frac{ d^2(d+2)^2 (d-1) }{4N}. 
\end{align}

\subsection{Proof of Theorem~\ref{theorem:extended_C-cost} and Corollary~\ref{thm:equiangular}}\label{sec:ecost_theorems}

\newtheorem{lemma}{Lemma}

In this section, we first present useful lemmas to prove Theorem~\ref{theorem:extended_C-cost} and its Corollary~\ref{thm:equiangular} that allow for analytical calculation of the optimal bound of the C-cost. 

The first lemma is trivial from the singular-value decomposition of a matrix $A = U\Sigma V^T$, where $U, V$ are orthogonal matrices, and $\Sigma$ is the diagonal matrix that contains the singular values of $A$. 
\begin{lemma}\label{lemma:XX_eigenvalues}
Let $A$ be a real matrix. The multiset of non-zero eigenvalues of $AA^T$ is the same as the multiset of non-zero eigenvalues of $A^TA$. 
\end{lemma}

\begin{lemma}\label{lemma:identity_ext}
Let $A$ be a real symmetric matrix such that one of its eigenvalues is $a$ and the rest are $b$'s. Then, it holds that $A = (a-b) \bm{u}\bm{u}^T + b I$ where $\bm{u}$ is the (normalized) eigenvector corresponding to the eigenvalue $a$.
\end{lemma}

\begin{proof}
Easy by seeing that $A\bm{u} = a \bm{u}$, and $A\bm{v} = b \bm{v}$ hold for every $\bm{v}$ which is orthogonal to $\bm{u}$, i.e., $\bm{v}^T \bm{u} = 0$. 
\end{proof}

Let $A$ be an $n\times n$ positive definite matrix with the largest eigenvalue $\lambda_{\max}$ and the smallest eigenvalue $\lambda_{\min}$ such that $\kappa = \lambda_{\max} / \lambda_{\min}$. It is known that $n^2/\kappa \le \mbox{Tr}\left(A\right)~\mbox{Tr}\left(A^{-1}\right) \le n^2 \kappa$ holds with equality if and only if $\kappa = 1$, i.e., $A = \lambda I$ for some $\lambda > 0$. We formalize this in the following lemma. 
\begin{lemma}\label{lemma:trace_and_inverse}
Any positive-definite real symmetric matrix $A\in \mathbb{R}^{n\times n}$ satisfies $\Tr{A^{-1}} \geq n^2\Tr{A}^{-1}$ with equality if and only if $A= \lambda I$ for $\lambda > 0$.
\end{lemma}

We now prove Theorem~\ref{theorem:extended_C-cost} and its Corollary~\ref{thm:equiangular} concerning lower bounds and its equality conditions for C-cost. 
Here we revisit Theorem~\ref{theorem:extended_C-cost} for convenience.

\begin{theomain}
Suppose a single-qubit gate expressed by a parameter $\bm{q}$ in $\mathbb{R}^{d}$ where $|\bm{q}|=1$.
Let $\{\bm{q}_1, \cdots, \bm{q}_N \}$ be a parameter configuration and let $A$ be the corresponding matrix $A=[h(q_1),\cdots,h(q_N)]^T$ in $\mathbb{R}^{N\times N_\mathrm{min}}$ where $N\ge N_\mathrm{ min}\equiv d(d+1)/2$.
The C-cost $C(A)$ defined as
\begin{align} \label{eq:detailed_ex_cost}
    C(A) 
    &=\frac{N}{N_{\mathrm{min}}d(d+2)}
    {\rm Tr}[ (A^TA)^{-1}(\bm{1}_d\bm{1}_d^T+2I)]
\end{align}
satisfies $C(A) \geq 1$ 
with equality if and only if the parameter configuration $\{ \bm{q}_i \}$ and $A$ satisfy
\begin{equation}
    A^TA = \frac{N}{d(d+2)} (\bm{1}_d\bm{1}_d^T + 2I).
\end{equation}

\end{theomain}

\begin{proof}
Using the Woodbury matrix identity giving
\begin{align}\label{eq:targetval}
    \left( \bm{1}_d \bm{1}_d^T + 2I \right)^{-1} =  \left( \frac{1}{2}I - \frac{1}{2(d+2)} \bm{1}_d \bm{1}_d^T \right),
\end{align}
we obtain the lower bound of Eq.~\eqref{eq:targetval} as
\begin{eqnarray} \label{eq:WW_lb}
    && \Tr{ (A^T A)^{-1}(\bm{1}_d\bm{1}_d^T + 2I) } \nonumber \\
    &=& \Tr{ \left( \left(\bm{1}_d\bm{1}_d^T + 2I \right)^{-1} (A^T A)\right)^{-1} } \nonumber \\
    &=& \Tr{ \left( \frac{1}{2}A^T A - \frac{1}{2(d+2)} \bm{1}_d \bm{1}_d^T  A^T A\right)^{-1} } \nonumber \\
    &\geq& N_\mathrm{min}^2 \Tr{ \frac{1}{2}A^T A - \frac{1}{2(d+2)} \bm{1}_d \bm{1}_d^T A^T A}^{-1} \nonumber \\
    &=& \frac{N_\mathrm{min} d(d+2)}{N}
\end{eqnarray}
where the inequality in the fourth line is derived by Lemma~\ref{lemma:trace_and_inverse}.
To obtain the last line, we use $\Tr{A^T A} = \Tr{A A^T} = N$ and $\Tr{\bm{1}_d \bm{1}_d^T A^T A } = N$ as well as
 $N_\mathrm{min} = d(d+1)/2$.
Therefore, $C(A)\ge 1$.

According to Lemma~\ref{lemma:trace_and_inverse}, the equality in the fourth line in Eq.~\eqref{eq:WW_lb} is given as
\begin{align}
\label{eq:equality}
    \frac{1}{2}A^T A - \frac{1}{2(d+2)} \bm{1}_d \bm{1}_d^T A^T A
    = \lambda I,
\end{align}
where $\lambda$ is a constant.
Tracing over both sides of Eq.~\eqref{eq:equality}, we have
\begin{align}
    \lambda 
    = \frac{N}{d(d+2)}.
\end{align}
Therefore, $C(A)=1$ if and only if 
\begin{equation}\label{eq:proof_phase_1}
    A^T A = \frac{N}{d(d+2)} \left(\bm{1}_d \bm{1}_d^T + 2I \right).
\end{equation}
\end{proof}

\textit{
\begin{corollary}
    For the minimum number of parameters ($N=N_\mathrm{min}$),  
    it holds that $C(A) \geq 1$ with equality if and only if the parameter configurations $\{\bm{q}_i\}_{i=1}^{N}$ satisfy
\begin{equation} \label{eq:general_equidistance}
    |\bm{q}_i\cdot\bm{q}_{j}| = \frac{1}{\sqrt{d+2}}
     ~ (\mathrm{for ~all} ~ i\neq j).
\end{equation}    
\end{corollary}
}
\begin{proof}
We show
Eq.~\eqref{eq:proof_phase_1} is equivalent to Eq.\eqref{eq:general_equidistance} if $N=N_\mathrm{min}$.
We first show
\begin{align}
    A^T A &= \frac{N_\mathrm{min}}{d(d+2)} \left(\bm{1}_d \bm{1}_d^T + 2I \right) \notag\\ 
    \Longrightarrow ~~~
    |\bm{q}_i\cdot\bm{q}_{j}| &= \frac{1}{ \sqrt{(d+2)}}
     ~ (\mathrm{for ~all} ~ i\neq j) \notag.
\end{align}
Recall that  
$A = [\bm{h}(\bm{q}_1), \cdots, \bm{h}(\bm{q}_N) ]^T$.
If $N=N_\mathrm{min}$, both $AA^T$ and $A^TA$ lie in $\mathbb{R}^{N_\mathrm{min}\times N_\mathrm{min}}$,
$A^TA=\sum_i^{N_\mathrm{min}} \bm{h}(\bm{q}_i) \bm{h}(\bm{q}_i)^T$ 
and so
\begin{equation} \label{eq:hhT}
    \sum_{i=1}^{N_\mathrm{min}} \bm{h}(\bm{q}_i) \bm{h}(\bm{q}_i)^T =\frac{N_\mathrm{min}}{d(d+2)} \left(\bm{1}_d \bm{1}_d^T + 2I \right)
\end{equation}
Multiplying both sides by $\bm{1}_d$ from the right, we obtain  
\begin{equation}\label{eq:sum_h}
    \sum_{i=1}^{N_\mathrm{min}} \bm{h}(\bm{q}_i) = \frac{N_\mathrm{min}}{d} \bm{1}_d,
\end{equation}
According to Lemma~\ref{lemma:XX_eigenvalues}, $A^T A$ and $AA^T$ have the identical set of non-zero eigenvalues, i.e. one of eigenvalues is $d+2$ and the rest are 2.
Then, using Lemma~\ref{lemma:identity_ext},
$A A^T$ can be expressed as
\begin{equation} \label{eq:AAT}
    A A^T = \frac{N_\mathrm{min}}{d(d+2)} \left(d \bm{v} \bm{v}^T + 2I \right),
\end{equation}
where $\bm{v}\in\mathbb{R}^{N_\mathrm{min}}$ is a unit vector.
On the other hand, 
the $(i,j)$-component of $AA^T$ has a relation 
\begin{equation} \label{eq:hTh}
    (AA^T)_{ij}=\bm{h}(\bm{q}_i)^T \bm{h}(\bm{q}_j) 
    = \frac{N_\mathrm{min}}{d(d+2)} \left(d v_i v_j + 2 \delta_{ij} \right).
\end{equation}
Summing Eq.~\eqref{eq:hTh} over $j$ from 1 to $N_\mathrm{min}$ and using $\bm{h}(\bm{q})^T \bm{1}_d=1$ and Eq.~\eqref{eq:sum_h},  we obtain
\begin{equation} \label{eq:v}
    v_j = \frac{1}{\sum_i^{N_\mathrm{min}} v_i}.
\end{equation}
Since $\bm{v}$ is a unit vector, $\bm{v} = \pm \bm{1}_{N_\mathrm{min}} / \sqrt{N_\mathrm{min}}$, where $\bm{1}_{N_\mathrm{min}} \in \mathbb{R}^{N_\mathrm{min}}$ is a vector whose all elements are 1.
Therefore, 
\begin{align} \label{eq:hTh_condition}
    \bm{h}(\bm{q}_i)^T \bm{h}(\bm{q}_j) = \frac{1}{d+2}, \text{ for } i \neq j.
\end{align}
Using the relation $\bm{h}(\bm{q}_i)^T \bm{h}(\bm{q}_j) = (\bm{q}_i \cdot \bm{q}_j)^2$, we obtain
\begin{equation}
    | \bm{q}_i \cdot \bm{q}_{j(\neq i)} | = \frac{1}{\sqrt{d+2}}. 
\end{equation}

Next, we prove that 
\begin{align}
    |\bm{q}_i\cdot\bm{q}_{j}| 
    &= \frac{1}{ \sqrt{(d+2)}}
     ~ (\mathrm{for ~all} ~ i\neq j) \notag    \\
     \Longrightarrow ~~~
    A^T A &=     
    \frac{N}{d(d+2)} 
    \left(\bm{1}_d \bm{1}_d^T + 2I \right) \notag.
\end{align}
Using the relation $\bm{h}(\bm{q}_i)^T \bm{h}(\bm{q}_j) = (\bm{q}_i \cdot \bm{q}_j)^2$ and $|\bm{q}_i|^2=1$ again, we obtain
\begin{align} \label{eq:hhT_converse}
    \bm{h}(\bm{q}_i)^T \bm{h}(\bm{q}_j) = \frac{1}{d+2} + \frac{d+1}{d+2} \delta_{ij}.
\end{align}
Since $\bm{h}(\bm{q}_i)^T \bm{h}(\bm{q}_j)$ is the $(i, j)$-component of $A A^T$, we can write
\begin{align} \label{eq:AAT_converse}
    A A^T = \frac{N_\mathrm{min}}{d(d+2)} \left( \frac{d}{N_\mathrm{min}} \bm{1}_{N_\mathrm{min}} \bm{1}_{N_\mathrm{min}}^T + 2I \right),
\end{align}
because of $N_\mathrm{min} = d(d+1)/2$.
Using Eq.~\eqref{eq:AAT_converse}, Lemmas~\ref{lemma:XX_eigenvalues} and \ref{lemma:identity_ext}, we can write $A^T A$ as
\begin{equation} \label{eq:ATA_converse}
    A^T A = \frac{N_\mathrm{min}}{d(d+2)} \left(d \bm{v}' \bm{v}'^T + 2I \right),
\end{equation}
where $v'=[v'_1, \ldots, v'_{N_\mathrm{min}}]^T$ is a unit vector.
Since $A^T A = \sum_i \bm{h}(\bm{q}_i) \bm{h}(\bm{q}_i)^T$ and $\forall i, \bm{h}(\bm{q}_i)^T \bm{1}_d = 1$, multiplying Eq.~\eqref{eq:ATA_converse} by $\bm{1}_d$ from the right side yields
\begin{align} \label{eq:sum_h_converse}
    \sum_i \bm{h}(\bm{q}_i) = \frac{N_\mathrm{min}}{d(d+2)}\left( d \bm{v}' \bm{v}'^T \bm{1}_d + 2 \bm{1}_d \right).
\end{align}
Summing Eq.~\eqref{eq:hhT_converse} over $i$ and $j$ from 1 to $N_\mathrm{min}$, we obtain
\begin{align} \label{eq:sum_hhT_converse}
    \sum_{i, j} \bm{h}(\bm{q}_i)^T \bm{h}(\bm{q}_j) = \frac{N_\mathrm{min}^2}{d},
\end{align}
which further yields
\begin{align} \label{eq:v_converse}
    (\bm{1}_d^T \bm{v}')^2 = d,
\end{align}
by substituting Eq.~\eqref{eq:sum_h_converse} into Eq.~\eqref{eq:sum_hhT_converse} and rearranging the resultant equation.
Since $\| \bm{1}_d \| = \sqrt{d}$, Eq.~\eqref{eq:v_converse} means $\bm{v}' = \pm \bm{1}_d / \sqrt{d}$.
Therefore, Eq.~\eqref{eq:ATA_converse} becomes
\begin{equation}
    A^T A = \frac{N_\mathrm{min}}{d(d+2)} \left( \bm{1}_d \bm{1}_d^T + 2I \right),
\end{equation}
which is just the equality conditon of $C(A) = 1$.

\end{proof}

\subsection{Proof of rotation invariance of C-cost}\label{apdx:rotation_invariance}
Here, we prove the C-cost $C(A)$ is invariant with respect to the parameter rotations as 
\begin{equation}
    C(A) = C(A_R),
\end{equation}
where the subscript $R$ stands for the rotated parameter set.
Let $\{ \bm{q}_1, \bm{q}_2, ..., \bm{q}_N \}$ be the original parameter configuration.
Then, a rotation matrix $R \in SO(d)$ gives another parameter configutation $\{ R\bm{q}_1, R\bm{q}_2, ..., R\bm{q}_N \}$.
For convenience, we define a matrix $Q$ as\begin{equation} \label{eq:Q-matrix}
    Q = ( \bm{q}_1 \otimes \bm{q}_1, \bm{q}_2 \otimes \bm{q}_2, ..., \bm{q}_N \otimes \bm{q}_N )^T.
\end{equation}
Likewise,
\begin{eqnarray}
    Q_R &=& ( (R\bm{q}_1) \otimes (R\bm{q}_1), ..., (R\bm{q}_N) \otimes (R\bm{q}_N) )^T \notag \\
    &=& Q (R^T \otimes R^T).
\end{eqnarray}
Using Eq.~\eqref{eq:x:hmql} in Appendix \ref{apdx:analytical_form}, $Q$ is linked to $A$ as
\begin{equation}
    AL^T = Q ~\mathrm{and}~ A=QL.
\end{equation}
which implies $Q$ encodes the parameter configurations as well as $A$.
Thus, the matrix $A$ for the rotated parameter set is given as
\begin{equation}
    A_R L^T = Q(R^T \otimes R^T) ~\mathrm{and}~ A_R=Q(R^T \otimes R^T)L.
\end{equation}
From Eq.~\eqref{eq:ZiiZss} and \eqref{eq:ZijZij}, the C-cost contains the Gram matrix $A^TA$.
For the rotated parameter set, the corresponding Gram matrix is given as
\begin{eqnarray}
    A_R^TA_R &=& L^T (R\otimes R) Q^T Q (R^T\otimes R^T) L \\
    &=& L^T (R\otimes R) L A^T A L^T (R^T\otimes R^T) L \\
    &=& R_{L} A^TA R_{L}^T,
\end{eqnarray}
where we denote $R_{L}:=L^T (R\otimes R) L$.

In fact, the first and second terms of Eq.~\eqref{eq:x:use_orthogo} are independently invariant for parameter rotations as follows.

For the first term $\bm{1}_d (A^TA)^{-1} \bm{1}_d^T$ (Eq.~\eqref{eq:ZiiZss}), the rotated version of the first term is expanded as
\begin{eqnarray}
    \bm{1}_d^T (A^T_RA_R)^{-1} \bm{1}_d 
    &=& \bm{1}_d^T (R_{L} A^TA R_{L}^T)^{-1} \bm{1}_d \notag \\
    &=& (R_{L}^{-1} \bm{1}_d)^T (A^TA)^{-1} R_{L}^{-1} \bm{1}_d \notag \\
    &=& \bm{1}_d^T (A^TA)^{-1} \bm{1}_d,
\end{eqnarray}
where we use the fact $R_{L}^{-1} \bm{1}_d = \bm{1}_d$, which is easily derived as
\begin{eqnarray}
    R_{L}\bm{1}_d
    &=& L^T(R\otimes R) LL^T \mbox{vec}(I) \notag \\
    &=& L^T LL^T (R\otimes R) \mbox{vec}(I) \notag \\
    &=& L^T (R\otimes R) \mbox{vec}(I) \notag \\
    &=& L^T \mbox{vec}(R I R^T) \notag \\
    &=& L^T \mbox{vec}(I) \notag \\
    &=& \bm{1}_d,
\end{eqnarray}
where we employed $LL^T(R\otimes R)=(R\otimes R)LL^T$ and $L^TL=I$. 
This equations implies the first term is rotation invariant.

For the second term ${\rm Tr}[ (A^TA)^{-1} ]$ (Eq.~\eqref{eq:ZijZij}), the rotated version of the second term is expanded as
\begin{eqnarray}
    {\rm Tr}[ (A^T_RA_R)^{-1} ]
    &=& {\rm Tr}[  (R_{L} A^TA R_{L}^T)^{-1} ] \notag \\
    &=& {\rm Tr}[ (A^TAR_{L}^TR_{L})^{-1} ] \notag \\
    &=& {\rm Tr}[ (A^TA)^{-1} ]
\end{eqnarray}
where for the second equality $I = R_{L}^TR_{L}$ is employed, which is derived as
\begin{align}
    R_{L}^TR_{L} 
    =& L (R^T\otimes R^T) L L^T (R \otimes R) L \notag \\
    =& L^T L L^T(R\otimes R)(R^T \otimes R^T) L \notag \\
    =& L^T L L^T((RR^T)\otimes (RR^T)) L \notag \\
    =&  L^T L L^T L \notag \\
    =& I,
\end{align}
where we employed $LL^T(R\otimes R)=(R\otimes R)LL^T$ and $L^TL=I$.
This equations implies the second term is rotation invariant.
Consequently, the C-cost is rotation invariant, because both the two terms in the C-cost are rotation invariant. ($\mathbb{E}[\Delta E]$  is also rotation invariant, because it is the weighted sum of these two terms.) $\square$

\begin{figure*}[t]\includegraphics[width=0.9\linewidth]{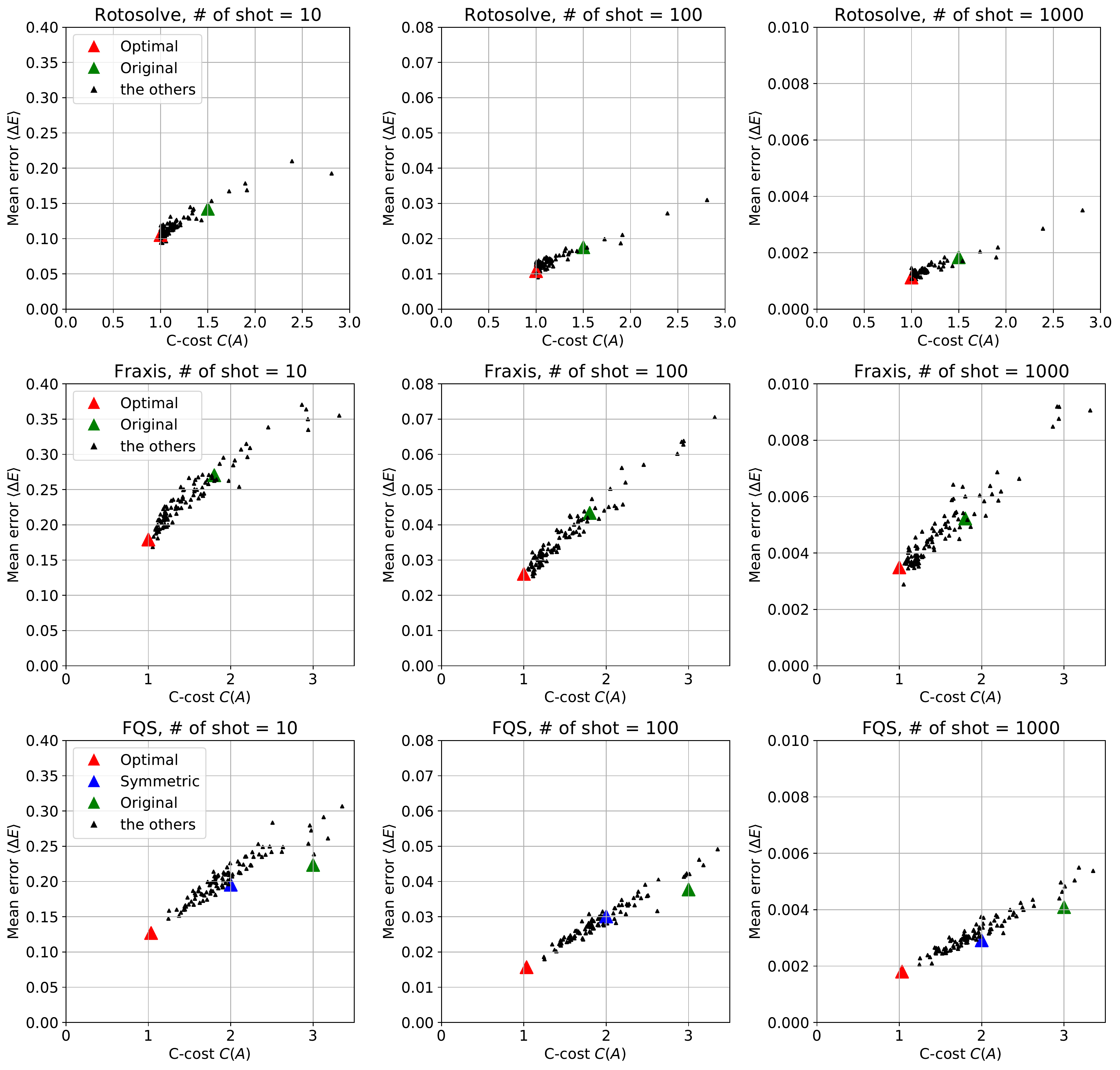}
    \caption{ \textbf{Additional results of the average energy error for the one-time optimization on different target gates in the ansatz with various parameter configurations.} }
    \label{ fig:appdxF}
\end{figure*}

\subsection{Comparison of our Parameters with the original methods.}\label{apdx:used_parameter}
We show the parameter values used as a sequential optimization in the main text as follows. The parameters are in no particular order.

\subsubsection{Rx gate}
The original parameter configuration for Rx gate 
$\bm{r}_1, \bm{r}_2, \bm{r}_3$ proposed in \cite{ostaszewski2021} is represented as
\begin{equation}
    \begin{bmatrix}
        \bm{r}^T_1 \\ \bm{r}^T_2 \\ \bm{r}^T_3
    \end{bmatrix}
     =
    \begin{bmatrix}
        1 & 0 \\
        \cos(\pi/4) & \sin(\pi/4) \\
        \cos(\pi/4) & -\sin(\pi/4) 
    \end{bmatrix}.
\end{equation}

The unique optimal parameter configuration for Rx gate with minimum number parameter set $\bm{r}_1, \bm{r}_2, \bm{r}_3$ is analytically derived as 
\begin{equation}
    \begin{bmatrix}
        \bm{r}^T_1 \\ \bm{r}^T_2 \\ \bm{r}^T_3
    \end{bmatrix}
     =
    \begin{bmatrix}
        1 & 0 \\
        \cos(\pi/3) & \sin(\pi/3) \\
        \cos(\pi/3) & -\sin(\pi/3) 
    \end{bmatrix}
\end{equation}
and its arbitrary rotation and (individual) reversal.

\subsubsection{Fraxis gate}
The original parameter configuration for Fraxis gate 
$\bm{n}_1, \bm{n}_2, ..., \bm{n}_6$ proposed in \cite{watanabe2021} is represented as
\begin{equation}
    \label{eq:orig_MC_Fraxis}
    \begin{bmatrix}
        \bm{n}^T_1 \\ \bm{n}^T_2 \\
        \bm{n}^T_3 \\ \bm{n}^T_4 \\
        \bm{n}^T_5 \\ \bm{n}^T_6
    \end{bmatrix}
     = \frac{1}{\sqrt{2}}
    \begin{bmatrix}
        \sqrt{2} & 0 & 0 \\
        0 & \sqrt{2} & 0 \\
        0 & 0 & \sqrt{2} \\
        1 & 1 & 0 \\
        1 & 0 & 1 \\
        0 & 1 & 1
    \end{bmatrix}.
\end{equation}

The unique (up to arbitrary rotation and individual reversal) optimal parameter configuration for Fraxis gate with minimum number parameter set $\bm{n}_1, \bm{n}_2, ..., \bm{n}_6$ is analytically derived as the vertices of icosahedron  
\begin{equation}
    \label{eq:opt_MC_Fraxis}
    \begin{bmatrix}
        \bm{n}^T_1 \\ \bm{n}^T_2 \\
        \bm{n}^T_3 \\ \bm{n}^T_4 \\
        \bm{n}^T_5 \\ \bm{n}^T_6
    \end{bmatrix}
     =
     \frac{1}{\sqrt{1+\varphi^2}}
     \begin{bmatrix}
     0 & 1 & \varphi\\
     0 & 1 & -\varphi\\
     1 & \varphi & 0 \\
     1 & -\varphi& 0\\
     \varphi & 0 & 1\\
     -\varphi& 0 & 1
     \end{bmatrix},
\end{equation}
where $\varphi = \frac{1+\sqrt{5}}{2}$ is the golden ratio. 

\subsubsection{FQS gate}
The original parameter configuration for FQS 
$\bm{q}_1, \bm{q}_2, ..., \bm{q}_{10}$ proposed in \cite{wada2022full} is represented as
\begin{equation}
    \label{eq:orig_MC_FQS}
    \begin{bmatrix}
        \bm{q}^T_1 \\ \bm{q}^T_2 \\
        \bm{q}^T_3 \\ \bm{q}^T_4 \\
        \bm{q}^T_5 \\ \bm{q}^T_6 \\
        \bm{q}^T_7 \\ \bm{q}^T_8 \\
        \bm{q}^T_9 \\ \bm{q}^T_{10}
    \end{bmatrix}
     =
     \frac{1}{\sqrt{2}}
    \begin{bmatrix}
        \sqrt{2} & 0 & 0 & 0 \\
        1 & -1 & 0 & 0 \\
        1 & 0 & -1 & 0 \\
        1 & 0 & 0 & -1 \\
        1 & 1 & 0 & 0 \\
        1 & 0 & 1 & 0 \\
        1 & 0 & 0 & 1 \\
        0 & 1 & 1 & 0 \\
        0 & 1 & 0 & 1 \\
        0 & 0 & 1 & 1 \\
    \end{bmatrix}.
\end{equation}

The symmetric parameter configuration for FQS gate 
$\bm{q}_1, \bm{q}_2, ..., \bm{q}_{10}$ which is only used for the experimental results in main text is represented as
\begin{equation}
    \label{eq:sym_MC_FQS}
    \begin{bmatrix}
        \bm{q}^T_1 \\ \bm{q}^T_2 \\
        \bm{q}^T_3 \\ \bm{q}^T_4 \\
        \bm{q}^T_5 \\ \bm{q}^T_6 \\
        \bm{q}^T_7 \\ \bm{q}^T_8 \\
        \bm{q}^T_9 \\ \bm{q}^T_{10}
    \end{bmatrix}
     =
     \frac{1}{\sqrt{2}}
    \begin{bmatrix}
        \sqrt{2} & 0 & 0 & 0 \\
        0 & \sqrt{2} & 0 & 0 \\
        0 & 0 & \sqrt{2} & 0 \\
        0 & 0 & 0 & \sqrt{2} \\
        1 & 1 & 0 & 0 \\
        1 & 0 & 1 & 0 \\
        1 & 0 & 0 & 1 \\
        0 & 1 & 1 & 0 \\
        0 & 1 & 0 & 1 \\
        0 & 0 & 1 & 1 \\
    \end{bmatrix}.
\end{equation}

The unique optimal parameter configuration for FQS gate with minimum number parameter set $\bm{q}_1, \bm{q}_2, ..., \bm{q}_{10}$ is numerically derived as 
\begin{equation}
    \label{eq:opt_MC_FQS}
    \begin{bmatrix}
        \bm{q}^T_1 \\ \bm{q}^T_2 \\
        \bm{q}^T_3 \\ \bm{q}^T_4 \\
        \bm{q}^T_5 \\ \bm{q}^T_6 \\
        \bm{q}^T_7 \\ \bm{q}^T_8 \\
        \bm{q}^T_9 \\ \bm{q}^T_{10}
    \end{bmatrix}
     =
    \begin{bmatrix}
        a & b & b & b \\
        b & a & b & b \\
        b & b & a & b \\
        b & b & b & a \\
        c & c & d & d \\
        c & d & c & d \\
        c & d & d & c \\
        d & c & c & d \\
        d & c & d & c \\
        d & d & c & c
    \end{bmatrix}
\end{equation}
and its arbitrary rotation and (individual) reversal, where $a=\sqrt{3}/2$, $b=-1/(2\sqrt{3})$, and $c^2 + d^2 = 1/2$, where $c \approx 0.7049, d \approx -0.0561$.

\subsection{\texorpdfstring{$\Delta E$}{ΔE} Distributions sampled with various parameter configurations}\label{apdx:Distribution}
In the experiment in Sec.~\ref{subsec:expB}, we performed optimization of only one gate to investigate the estimation error of the target gate. 
In the main text we shows only the case of $U_2$ (Ry gate in $U_2$ for Rotosolve case) as the target gate. 
In this section we show another case, that is the case of the target gate is $U_0$ for the FQS and Fraxis case and the Ry gate of $U_0$ for the Rotosolve case. 
Note that the number of shot per circuit $S$ is set to 10, 100, 1000 and the parameters of all the gates are initialized to random values and only the target gate is optimized. 
Fig.~\ref{ fig:appdxF} shows the results of all the additional experiments. The title of each subplot tells the target gate and other settings.

%% file: main.bbl
\begin{thebibliography}{40}%
\makeatletter
\providecommand \@ifxundefined [1]{%
 \@ifx{#1\undefined}
}%
\providecommand \@ifnum [1]{%
 \ifnum #1\expandafter \@firstoftwo
 \else \expandafter \@secondoftwo
 \fi
}%
\providecommand \@ifx [1]{%
 \ifx #1\expandafter \@firstoftwo
 \else \expandafter \@secondoftwo
 \fi
}%
\providecommand \natexlab [1]{#1}%
\providecommand \enquote  [1]{``#1''}%
\providecommand \bibnamefont  [1]{#1}%
\providecommand \bibfnamefont [1]{#1}%
\providecommand \citenamefont [1]{#1}%
\providecommand \href@noop [0]{\@secondoftwo}%
\providecommand \href [0]{\begingroup \@sanitize@url \@href}%
\providecommand \@href[1]{\@@startlink{#1}\@@href}%
\providecommand \@@href[1]{\endgroup#1\@@endlink}%
\providecommand \@sanitize@url [0]{\catcode `\\12\catcode `\$12\catcode
  `\&12\catcode `\#12\catcode `\^12\catcode `\_12\catcode `\%12\relax}%
\providecommand \@@startlink[1]{}%
\providecommand \@@endlink[0]{}%
\providecommand \url  [0]{\begingroup\@sanitize@url \@url }%
\providecommand \@url [1]{\endgroup\@href {#1}{\urlprefix }}%
\providecommand \urlprefix  [0]{URL }%
\providecommand \Eprint [0]{\href }%
\providecommand \doibase [0]{https://doi.org/}%
\providecommand \selectlanguage [0]{\@gobble}%
\providecommand \bibinfo  [0]{\@secondoftwo}%
\providecommand \bibfield  [0]{\@secondoftwo}%
\providecommand \translation [1]{[#1]}%
\providecommand \BibitemOpen [0]{}%
\providecommand \bibitemStop [0]{}%
\providecommand \bibitemNoStop [0]{.\EOS\space}%
\providecommand \EOS [0]{\spacefactor3000\relax}%
\providecommand \BibitemShut  [1]{\csname bibitem#1\endcsname}%
\let\auto@bib@innerbib\@empty
\bibitem [{\citenamefont {Peruzzo~\textit{et al}}(2014)}]{Peruzzo2014NatCom}%
  \BibitemOpen
  \bibfield  {author} {\bibinfo {author} {\bibfnamefont {A.}~\bibnamefont
  {Peruzzo~\textit{et al}}},\ }\href {https://doi.org/10.1038/ncomms5213}
  {\bibfield  {journal} {\bibinfo  {journal} {Nat. Commun.}\ }\textbf {\bibinfo
  {volume} {5}},\ \bibinfo {pages} {4213} (\bibinfo {year} {2014})}\BibitemShut
  {NoStop}%
\bibitem [{\citenamefont {Kandala~\textit{et al}}(2017)}]{Kandala2017Nat}%
  \BibitemOpen
  \bibfield  {author} {\bibinfo {author} {\bibfnamefont {A.}~\bibnamefont
  {Kandala~\textit{et al}}},\ }\href {https://doi.org/10.1038/nature23879}
  {\bibfield  {journal} {\bibinfo  {journal} {Nature}\ }\textbf {\bibinfo
  {volume} {549}},\ \bibinfo {pages} {242} (\bibinfo {year}
  {2017})}\BibitemShut {NoStop}%
\bibitem [{\citenamefont {Tilly}\ \emph {et~al.}(2022)\citenamefont {Tilly},
  \citenamefont {Chen}, \citenamefont {Cao}, \citenamefont {Picozzi},
  \citenamefont {Setia}, \citenamefont {Li}, \citenamefont {Grant},
  \citenamefont {Wossnig}, \citenamefont {Rungger}, \citenamefont {Booth},\
  and\ \citenamefont {Tennyson}}]{TILLY20221}%
  \BibitemOpen
  \bibfield  {author} {\bibinfo {author} {\bibfnamefont {J.}~\bibnamefont
  {Tilly}}, \bibinfo {author} {\bibfnamefont {H.}~\bibnamefont {Chen}},
  \bibinfo {author} {\bibfnamefont {S.}~\bibnamefont {Cao}}, \bibinfo {author}
  {\bibfnamefont {D.}~\bibnamefont {Picozzi}}, \bibinfo {author} {\bibfnamefont
  {K.}~\bibnamefont {Setia}}, \bibinfo {author} {\bibfnamefont
  {Y.}~\bibnamefont {Li}}, \bibinfo {author} {\bibfnamefont {E.}~\bibnamefont
  {Grant}}, \bibinfo {author} {\bibfnamefont {L.}~\bibnamefont {Wossnig}},
  \bibinfo {author} {\bibfnamefont {I.}~\bibnamefont {Rungger}}, \bibinfo
  {author} {\bibfnamefont {G.~H.}\ \bibnamefont {Booth}},\ and\ \bibinfo
  {author} {\bibfnamefont {J.}~\bibnamefont {Tennyson}},\ }\href
  {https://doi.org/https://doi.org/10.1016/j.physrep.2022.08.003} {\bibfield
  {journal} {\bibinfo  {journal} {Physics Reports}\ }\textbf {\bibinfo {volume}
  {986}},\ \bibinfo {pages} {1} (\bibinfo {year} {2022})}\BibitemShut {NoStop}%
\bibitem [{\citenamefont {Cerezo}\ \emph
  {et~al.}(2021{\natexlab{a}})\citenamefont {Cerezo}, \citenamefont
  {Arrasmith}, \citenamefont {Babbush}, \citenamefont {Benjamin}, \citenamefont
  {Endo}, \citenamefont {Fujii}, \citenamefont {McClean}, \citenamefont
  {Mitarai}, \citenamefont {Yuan}, \citenamefont {Cincio} \emph
  {et~al.}}]{cerezo2021variational}%
  \BibitemOpen
  \bibfield  {author} {\bibinfo {author} {\bibfnamefont {M.}~\bibnamefont
  {Cerezo}}, \bibinfo {author} {\bibfnamefont {A.}~\bibnamefont {Arrasmith}},
  \bibinfo {author} {\bibfnamefont {R.}~\bibnamefont {Babbush}}, \bibinfo
  {author} {\bibfnamefont {S.~C.}\ \bibnamefont {Benjamin}}, \bibinfo {author}
  {\bibfnamefont {S.}~\bibnamefont {Endo}}, \bibinfo {author} {\bibfnamefont
  {K.}~\bibnamefont {Fujii}}, \bibinfo {author} {\bibfnamefont {J.~R.}\
  \bibnamefont {McClean}}, \bibinfo {author} {\bibfnamefont {K.}~\bibnamefont
  {Mitarai}}, \bibinfo {author} {\bibfnamefont {X.}~\bibnamefont {Yuan}},
  \bibinfo {author} {\bibfnamefont {L.}~\bibnamefont {Cincio}}, \emph
  {et~al.},\ }\href@noop {} {\bibfield  {journal} {\bibinfo  {journal} {Nature
  Reviews Physics}\ }\textbf {\bibinfo {volume} {3}},\ \bibinfo {pages} {625}
  (\bibinfo {year} {2021}{\natexlab{a}})}\BibitemShut {NoStop}%
\bibitem [{\citenamefont {Huang}\ \emph {et~al.}(2022)\citenamefont {Huang},
  \citenamefont {Xu}, \citenamefont {Guo}, \citenamefont {Tian}, \citenamefont
  {Wei}, \citenamefont {Sun}, \citenamefont {Bao},\ and\ \citenamefont
  {Long}}]{huang2022near}%
  \BibitemOpen
  \bibfield  {author} {\bibinfo {author} {\bibfnamefont {H.-L.}\ \bibnamefont
  {Huang}}, \bibinfo {author} {\bibfnamefont {X.-Y.}\ \bibnamefont {Xu}},
  \bibinfo {author} {\bibfnamefont {C.}~\bibnamefont {Guo}}, \bibinfo {author}
  {\bibfnamefont {G.}~\bibnamefont {Tian}}, \bibinfo {author} {\bibfnamefont
  {S.-J.}\ \bibnamefont {Wei}}, \bibinfo {author} {\bibfnamefont
  {X.}~\bibnamefont {Sun}}, \bibinfo {author} {\bibfnamefont {W.-S.}\
  \bibnamefont {Bao}},\ and\ \bibinfo {author} {\bibfnamefont {G.-L.}\
  \bibnamefont {Long}},\ }\href@noop {} {\bibfield  {journal} {\bibinfo
  {journal} {arXiv preprint arXiv:2211.08737}\ } (\bibinfo {year}
  {2022})}\BibitemShut {NoStop}%
\bibitem [{\citenamefont {Gao~\textit{et al}}(2021)}]{gao2021applications}%
  \BibitemOpen
  \bibfield  {author} {\bibinfo {author} {\bibfnamefont {Q.}~\bibnamefont
  {Gao~\textit{et al}}},\ }\href {https://doi.org/10.1038/s41524-021-00540-6}
  {\bibfield  {journal} {\bibinfo  {journal} {npj Comput. Mater.}\ }\textbf
  {\bibinfo {volume} {7}},\ \bibinfo {pages} {1} (\bibinfo {year}
  {2021})}\BibitemShut {NoStop}%
\bibitem [{\citenamefont {Fuller}\ \emph {et~al.}(2021)\citenamefont {Fuller},
  \citenamefont {Hadfield}, \citenamefont {Glick}, \citenamefont {Imamichi},
  \citenamefont {Itoko}, \citenamefont {Thompson}, \citenamefont {Jiao},
  \citenamefont {Kagele}, \citenamefont {Blom-Schieber}, \citenamefont
  {Raymond} \emph {et~al.}}]{fuller2021approximate}%
  \BibitemOpen
  \bibfield  {author} {\bibinfo {author} {\bibfnamefont {B.}~\bibnamefont
  {Fuller}}, \bibinfo {author} {\bibfnamefont {C.}~\bibnamefont {Hadfield}},
  \bibinfo {author} {\bibfnamefont {J.~R.}\ \bibnamefont {Glick}}, \bibinfo
  {author} {\bibfnamefont {T.}~\bibnamefont {Imamichi}}, \bibinfo {author}
  {\bibfnamefont {T.}~\bibnamefont {Itoko}}, \bibinfo {author} {\bibfnamefont
  {R.~J.}\ \bibnamefont {Thompson}}, \bibinfo {author} {\bibfnamefont
  {Y.}~\bibnamefont {Jiao}}, \bibinfo {author} {\bibfnamefont {M.~M.}\
  \bibnamefont {Kagele}}, \bibinfo {author} {\bibfnamefont {A.~W.}\
  \bibnamefont {Blom-Schieber}}, \bibinfo {author} {\bibfnamefont
  {R.}~\bibnamefont {Raymond}}, \emph {et~al.},\ }\href
  {http://arxiv.org/abs/2111.03167} {\bibfield  {journal} {\bibinfo  {journal}
  {arXiv preprint arXiv:2111.03167}\ } (\bibinfo {year} {2021})}\BibitemShut
  {NoStop}%
\bibitem [{\citenamefont {Amaro}\ \emph {et~al.}(2022)\citenamefont {Amaro},
  \citenamefont {Rosenkranz}, \citenamefont {Fitzpatrick}, \citenamefont
  {Hirano},\ and\ \citenamefont {Fiorentini}}]{amaro2022case}%
  \BibitemOpen
  \bibfield  {author} {\bibinfo {author} {\bibfnamefont {D.}~\bibnamefont
  {Amaro}}, \bibinfo {author} {\bibfnamefont {M.}~\bibnamefont {Rosenkranz}},
  \bibinfo {author} {\bibfnamefont {N.}~\bibnamefont {Fitzpatrick}}, \bibinfo
  {author} {\bibfnamefont {K.}~\bibnamefont {Hirano}},\ and\ \bibinfo {author}
  {\bibfnamefont {M.}~\bibnamefont {Fiorentini}},\ }\href
  {https://doi.org/10.1140/epjqt/s40507-022-00123-4} {\bibfield  {journal}
  {\bibinfo  {journal} {EPJ Quantum Technology}\ }\textbf {\bibinfo {volume}
  {9}},\ \bibinfo {pages} {5} (\bibinfo {year} {2022})}\BibitemShut {NoStop}%
\bibitem [{\citenamefont {Zoufal}\ \emph {et~al.}(2022)\citenamefont {Zoufal},
  \citenamefont {Mishmash}, \citenamefont {Sharma}, \citenamefont {Kumar},
  \citenamefont {Sheshadri}, \citenamefont {Deshmukh}, \citenamefont {Ibrahim},
  \citenamefont {Gacon},\ and\ \citenamefont
  {Woerner}}]{zoufal2022variational}%
  \BibitemOpen
  \bibfield  {author} {\bibinfo {author} {\bibfnamefont {C.}~\bibnamefont
  {Zoufal}}, \bibinfo {author} {\bibfnamefont {R.~V.}\ \bibnamefont
  {Mishmash}}, \bibinfo {author} {\bibfnamefont {N.}~\bibnamefont {Sharma}},
  \bibinfo {author} {\bibfnamefont {N.}~\bibnamefont {Kumar}}, \bibinfo
  {author} {\bibfnamefont {A.}~\bibnamefont {Sheshadri}}, \bibinfo {author}
  {\bibfnamefont {A.}~\bibnamefont {Deshmukh}}, \bibinfo {author}
  {\bibfnamefont {N.}~\bibnamefont {Ibrahim}}, \bibinfo {author} {\bibfnamefont
  {J.}~\bibnamefont {Gacon}},\ and\ \bibinfo {author} {\bibfnamefont
  {S.}~\bibnamefont {Woerner}},\ }\href {http://arxiv.org/abs/2205.03045}
  {\bibfield  {journal} {\bibinfo  {journal} {arXiv preprint arXiv:2205.03045}\
  } (\bibinfo {year} {2022})}\BibitemShut {NoStop}%
\bibitem [{\citenamefont {Bravo-Prieto}\ \emph {et~al.}(2019)\citenamefont
  {Bravo-Prieto}, \citenamefont {LaRose}, \citenamefont {Cerezo}, \citenamefont
  {Subasi}, \citenamefont {Cincio},\ and\ \citenamefont
  {Coles}}]{bravo2019variational}%
  \BibitemOpen
  \bibfield  {author} {\bibinfo {author} {\bibfnamefont {C.}~\bibnamefont
  {Bravo-Prieto}}, \bibinfo {author} {\bibfnamefont {R.}~\bibnamefont
  {LaRose}}, \bibinfo {author} {\bibfnamefont {M.}~\bibnamefont {Cerezo}},
  \bibinfo {author} {\bibfnamefont {Y.}~\bibnamefont {Subasi}}, \bibinfo
  {author} {\bibfnamefont {L.}~\bibnamefont {Cincio}},\ and\ \bibinfo {author}
  {\bibfnamefont {P.~J.}\ \bibnamefont {Coles}},\ }\href
  {http://arxiv.org/abs/1909.05820} {\bibfield  {journal} {\bibinfo  {journal}
  {arXiv preprint arXiv:1909.05820}\ } (\bibinfo {year} {2019})}\BibitemShut
  {NoStop}%
\bibitem [{\citenamefont {Xu}\ \emph {et~al.}(2021)\citenamefont {Xu},
  \citenamefont {Sun}, \citenamefont {Endo}, \citenamefont {Li}, \citenamefont
  {Benjamin},\ and\ \citenamefont {Yuan}}]{xu2021variational}%
  \BibitemOpen
  \bibfield  {author} {\bibinfo {author} {\bibfnamefont {X.}~\bibnamefont
  {Xu}}, \bibinfo {author} {\bibfnamefont {J.}~\bibnamefont {Sun}}, \bibinfo
  {author} {\bibfnamefont {S.}~\bibnamefont {Endo}}, \bibinfo {author}
  {\bibfnamefont {Y.}~\bibnamefont {Li}}, \bibinfo {author} {\bibfnamefont
  {S.~C.}\ \bibnamefont {Benjamin}},\ and\ \bibinfo {author} {\bibfnamefont
  {X.}~\bibnamefont {Yuan}},\ }\href
  {https://doi.org/10.1016/j.scib.2021.06.023} {\bibfield  {journal} {\bibinfo
  {journal} {Science Bulletin}\ }\textbf {\bibinfo {volume} {66}},\ \bibinfo
  {pages} {2181} (\bibinfo {year} {2021})}\BibitemShut {NoStop}%
\bibitem [{\citenamefont {Sato}\ \emph {et~al.}(2021)\citenamefont {Sato},
  \citenamefont {Kondo}, \citenamefont {Koide}, \citenamefont {Takamatsu},\
  and\ \citenamefont {Imoto}}]{sato2021variational}%
  \BibitemOpen
  \bibfield  {author} {\bibinfo {author} {\bibfnamefont {Y.}~\bibnamefont
  {Sato}}, \bibinfo {author} {\bibfnamefont {R.}~\bibnamefont {Kondo}},
  \bibinfo {author} {\bibfnamefont {S.}~\bibnamefont {Koide}}, \bibinfo
  {author} {\bibfnamefont {H.}~\bibnamefont {Takamatsu}},\ and\ \bibinfo
  {author} {\bibfnamefont {N.}~\bibnamefont {Imoto}},\ }\href@noop {}
  {\bibfield  {journal} {\bibinfo  {journal} {Physical Review A}\ }\textbf
  {\bibinfo {volume} {104}},\ \bibinfo {pages} {052409} (\bibinfo {year}
  {2021})}\BibitemShut {NoStop}%
\bibitem [{\citenamefont {Sato}\ \emph {et~al.}(2023)\citenamefont {Sato},
  \citenamefont {Watanabe}, \citenamefont {Raymond}, \citenamefont {Kondo},
  \citenamefont {Wada}, \citenamefont {Endo}, \citenamefont {Sugawara},\ and\
  \citenamefont {Yamamoto}}]{sato2023variational}%
  \BibitemOpen
  \bibfield  {author} {\bibinfo {author} {\bibfnamefont {Y.}~\bibnamefont
  {Sato}}, \bibinfo {author} {\bibfnamefont {H.~C.}\ \bibnamefont {Watanabe}},
  \bibinfo {author} {\bibfnamefont {R.}~\bibnamefont {Raymond}}, \bibinfo
  {author} {\bibfnamefont {R.}~\bibnamefont {Kondo}}, \bibinfo {author}
  {\bibfnamefont {K.}~\bibnamefont {Wada}}, \bibinfo {author} {\bibfnamefont
  {K.}~\bibnamefont {Endo}}, \bibinfo {author} {\bibfnamefont {M.}~\bibnamefont
  {Sugawara}},\ and\ \bibinfo {author} {\bibfnamefont {N.}~\bibnamefont
  {Yamamoto}},\ }\href@noop {} {\bibfield  {journal} {\bibinfo  {journal}
  {arXiv preprint arXiv:2302.12602}\ } (\bibinfo {year} {2023})}\BibitemShut
  {NoStop}%
\bibitem [{\citenamefont {Cerezo}\ \emph
  {et~al.}(2021{\natexlab{b}})\citenamefont {Cerezo}, \citenamefont
  {Arrasmith}, \citenamefont {Babbush}, \citenamefont {Benjamin}, \citenamefont
  {Endo}, \citenamefont {Fujii}, \citenamefont {McClean}, \citenamefont
  {Mitarai}, \citenamefont {Yuan}, \citenamefont {Cincio},\ and\ \citenamefont
  {Coles}}]{Cerezo2021NatRevPhys}%
  \BibitemOpen
  \bibfield  {author} {\bibinfo {author} {\bibfnamefont {M.}~\bibnamefont
  {Cerezo}}, \bibinfo {author} {\bibfnamefont {A.}~\bibnamefont {Arrasmith}},
  \bibinfo {author} {\bibfnamefont {R.}~\bibnamefont {Babbush}}, \bibinfo
  {author} {\bibfnamefont {S.~C.}\ \bibnamefont {Benjamin}}, \bibinfo {author}
  {\bibfnamefont {S.}~\bibnamefont {Endo}}, \bibinfo {author} {\bibfnamefont
  {K.}~\bibnamefont {Fujii}}, \bibinfo {author} {\bibfnamefont {J.~R.}\
  \bibnamefont {McClean}}, \bibinfo {author} {\bibfnamefont {K.}~\bibnamefont
  {Mitarai}}, \bibinfo {author} {\bibfnamefont {X.}~\bibnamefont {Yuan}},
  \bibinfo {author} {\bibfnamefont {L.}~\bibnamefont {Cincio}},\ and\ \bibinfo
  {author} {\bibfnamefont {P.~J.}\ \bibnamefont {Coles}},\ }\href@noop {}
  {\bibfield  {journal} {\bibinfo  {journal} {Nature Reviews Physics 3,
  625-644}\ } (\bibinfo {year} {2021}{\natexlab{b}})}\BibitemShut {NoStop}%
\bibitem [{\citenamefont {Haghshenas}\ \emph {et~al.}(2022)\citenamefont
  {Haghshenas}, \citenamefont {Gray}, \citenamefont {Potter},\ and\
  \citenamefont {Chan}}]{GarnetChan2022PRX}%
  \BibitemOpen
  \bibfield  {author} {\bibinfo {author} {\bibfnamefont {R.}~\bibnamefont
  {Haghshenas}}, \bibinfo {author} {\bibfnamefont {J.}~\bibnamefont {Gray}},
  \bibinfo {author} {\bibfnamefont {A.~C.}\ \bibnamefont {Potter}},\ and\
  \bibinfo {author} {\bibfnamefont {G.~K.-L.}\ \bibnamefont {Chan}},\
  }\href@noop {} {\bibfield  {journal} {\bibinfo  {journal} {Physical Review X
  12, 011047}\ } (\bibinfo {year} {2022})}\BibitemShut {NoStop}%
\bibitem [{\citenamefont {Foss-Feig}\ \emph {et~al.}(2021)\citenamefont
  {Foss-Feig}, \citenamefont {Hayes}, \citenamefont {Dreiling}, \citenamefont
  {Figgatt}, \citenamefont {Gaebler}, \citenamefont {Moses}, \citenamefont
  {Pino},\ and\ \citenamefont {Potter}}]{foss2021holographic}%
  \BibitemOpen
  \bibfield  {author} {\bibinfo {author} {\bibfnamefont {M.}~\bibnamefont
  {Foss-Feig}}, \bibinfo {author} {\bibfnamefont {D.}~\bibnamefont {Hayes}},
  \bibinfo {author} {\bibfnamefont {J.~M.}\ \bibnamefont {Dreiling}}, \bibinfo
  {author} {\bibfnamefont {C.}~\bibnamefont {Figgatt}}, \bibinfo {author}
  {\bibfnamefont {J.~P.}\ \bibnamefont {Gaebler}}, \bibinfo {author}
  {\bibfnamefont {S.~A.}\ \bibnamefont {Moses}}, \bibinfo {author}
  {\bibfnamefont {J.~M.}\ \bibnamefont {Pino}},\ and\ \bibinfo {author}
  {\bibfnamefont {A.~C.}\ \bibnamefont {Potter}},\ }\href@noop {} {\bibfield
  {journal} {\bibinfo  {journal} {Physical Review Research}\ }\textbf {\bibinfo
  {volume} {3}},\ \bibinfo {pages} {033002} (\bibinfo {year}
  {2021})}\BibitemShut {NoStop}%
\bibitem [{\citenamefont {Haghshenas}(2021)}]{haghshenas2021optimization}%
  \BibitemOpen
  \bibfield  {author} {\bibinfo {author} {\bibfnamefont {R.}~\bibnamefont
  {Haghshenas}},\ }\href@noop {} {\bibfield  {journal} {\bibinfo  {journal}
  {Physical Review Research}\ }\textbf {\bibinfo {volume} {3}},\ \bibinfo
  {pages} {023148} (\bibinfo {year} {2021})}\BibitemShut {NoStop}%
\bibitem [{\citenamefont {Barratt}\ \emph {et~al.}(2021)\citenamefont
  {Barratt}, \citenamefont {Dborin}, \citenamefont {Bal}, \citenamefont
  {Stojevic}, \citenamefont {Pollmann},\ and\ \citenamefont
  {Green}}]{barratt2021parallel}%
  \BibitemOpen
  \bibfield  {author} {\bibinfo {author} {\bibfnamefont {F.}~\bibnamefont
  {Barratt}}, \bibinfo {author} {\bibfnamefont {J.}~\bibnamefont {Dborin}},
  \bibinfo {author} {\bibfnamefont {M.}~\bibnamefont {Bal}}, \bibinfo {author}
  {\bibfnamefont {V.}~\bibnamefont {Stojevic}}, \bibinfo {author}
  {\bibfnamefont {F.}~\bibnamefont {Pollmann}},\ and\ \bibinfo {author}
  {\bibfnamefont {A.~G.}\ \bibnamefont {Green}},\ }\href@noop {} {\bibfield
  {journal} {\bibinfo  {journal} {npj Quantum Information}\ }\textbf {\bibinfo
  {volume} {7}},\ \bibinfo {pages} {79} (\bibinfo {year} {2021})}\BibitemShut
  {NoStop}%
\bibitem [{\citenamefont {Liu}\ \emph {et~al.}(2019)\citenamefont {Liu},
  \citenamefont {Zhang}, \citenamefont {Wan},\ and\ \citenamefont
  {Wang}}]{liu2019variational}%
  \BibitemOpen
  \bibfield  {author} {\bibinfo {author} {\bibfnamefont {J.-G.}\ \bibnamefont
  {Liu}}, \bibinfo {author} {\bibfnamefont {Y.-H.}\ \bibnamefont {Zhang}},
  \bibinfo {author} {\bibfnamefont {Y.}~\bibnamefont {Wan}},\ and\ \bibinfo
  {author} {\bibfnamefont {L.}~\bibnamefont {Wang}},\ }\href@noop {} {\bibfield
   {journal} {\bibinfo  {journal} {Physical Review Research}\ }\textbf
  {\bibinfo {volume} {1}},\ \bibinfo {pages} {023025} (\bibinfo {year}
  {2019})}\BibitemShut {NoStop}%
\bibitem [{\citenamefont {Nakanishi}\ \emph {et~al.}(2020)\citenamefont
  {Nakanishi}, \citenamefont {Fujii},\ and\ \citenamefont
  {Todo}}]{nakanishi2020}%
  \BibitemOpen
  \bibfield  {author} {\bibinfo {author} {\bibfnamefont {K.~M.}\ \bibnamefont
  {Nakanishi}}, \bibinfo {author} {\bibfnamefont {K.}~\bibnamefont {Fujii}},\
  and\ \bibinfo {author} {\bibfnamefont {S.}~\bibnamefont {Todo}},\ }\href
  {https://doi.org/10.1103/physrevresearch.2.043158} {\bibfield  {journal}
  {\bibinfo  {journal} {Physical Review Research}\ }\textbf {\bibinfo {volume}
  {2}},\ \bibinfo {pages} {043158} (\bibinfo {year} {2020})}\BibitemShut
  {NoStop}%
\bibitem [{\citenamefont {Ostaszewski}\ \emph {et~al.}(2021)\citenamefont
  {Ostaszewski}, \citenamefont {Grant},\ and\ \citenamefont
  {Benedetti}}]{ostaszewski2021}%
  \BibitemOpen
  \bibfield  {author} {\bibinfo {author} {\bibfnamefont {M.}~\bibnamefont
  {Ostaszewski}}, \bibinfo {author} {\bibfnamefont {E.}~\bibnamefont {Grant}},\
  and\ \bibinfo {author} {\bibfnamefont {M.}~\bibnamefont {Benedetti}},\ }\href
  {https://doi.org/10.22331/q-2021-01-28-391} {\bibfield  {journal} {\bibinfo
  {journal} {Quantum}\ }\textbf {\bibinfo {volume} {5}},\ \bibinfo {pages}
  {391} (\bibinfo {year} {2021})}\BibitemShut {NoStop}%
\bibitem [{\citenamefont {Watanabe}\ \emph {et~al.}(2021)\citenamefont
  {Watanabe}, \citenamefont {Raymond}, \citenamefont {Ohnishi}, \citenamefont
  {Kaminishi},\ and\ \citenamefont {Sugawara}}]{watanabe2021}%
  \BibitemOpen
  \bibfield  {author} {\bibinfo {author} {\bibfnamefont {H.~C.}\ \bibnamefont
  {Watanabe}}, \bibinfo {author} {\bibfnamefont {R.}~\bibnamefont {Raymond}},
  \bibinfo {author} {\bibfnamefont {Y.-Y.}\ \bibnamefont {Ohnishi}}, \bibinfo
  {author} {\bibfnamefont {E.}~\bibnamefont {Kaminishi}},\ and\ \bibinfo
  {author} {\bibfnamefont {M.}~\bibnamefont {Sugawara}},\ }in\ \href
  {https://doi.org/https://doi.org/10.1109/QCE52317.2021.00026} {\emph
  {\bibinfo {booktitle} {2021 IEEE International Conference on Quantum
  Computing and Engineering (QCE)}}}\ (\bibinfo {organization} {IEEE},\
  \bibinfo {year} {2021})\ pp.\ \bibinfo {pages} {100--111}\BibitemShut
  {NoStop}%
\bibitem [{\citenamefont {Wada}\ \emph
  {et~al.}(2022{\natexlab{a}})\citenamefont {Wada}, \citenamefont {Raymond},
  \citenamefont {Ohnishi}, \citenamefont {Kaminishi}, \citenamefont {Sugawara},
  \citenamefont {Yamamoto},\ and\ \citenamefont {Watanabe}}]{wada2022}%
  \BibitemOpen
  \bibfield  {author} {\bibinfo {author} {\bibfnamefont {K.}~\bibnamefont
  {Wada}}, \bibinfo {author} {\bibfnamefont {R.}~\bibnamefont {Raymond}},
  \bibinfo {author} {\bibfnamefont {Y.-y.}\ \bibnamefont {Ohnishi}}, \bibinfo
  {author} {\bibfnamefont {E.}~\bibnamefont {Kaminishi}}, \bibinfo {author}
  {\bibfnamefont {M.}~\bibnamefont {Sugawara}}, \bibinfo {author}
  {\bibfnamefont {N.}~\bibnamefont {Yamamoto}},\ and\ \bibinfo {author}
  {\bibfnamefont {H.~C.}\ \bibnamefont {Watanabe}},\ }\href
  {https://doi.org/10.1103/PhysRevA.105.062421} {\bibfield  {journal} {\bibinfo
   {journal} {Physical Review A}\ }\textbf {\bibinfo {volume} {105}},\ \bibinfo
  {pages} {062421} (\bibinfo {year} {2022}{\natexlab{a}})}\BibitemShut
  {NoStop}%
\bibitem [{\citenamefont {Wada}\ \emph
  {et~al.}(2022{\natexlab{b}})\citenamefont {Wada}, \citenamefont {Raymond},
  \citenamefont {Sato},\ and\ \citenamefont {Watanabe}}]{wada2022full}%
  \BibitemOpen
  \bibfield  {author} {\bibinfo {author} {\bibfnamefont {K.}~\bibnamefont
  {Wada}}, \bibinfo {author} {\bibfnamefont {R.}~\bibnamefont {Raymond}},
  \bibinfo {author} {\bibfnamefont {Y.}~\bibnamefont {Sato}},\ and\ \bibinfo
  {author} {\bibfnamefont {H.~C.}\ \bibnamefont {Watanabe}},\ }\href@noop {}
  {\bibfield  {journal} {\bibinfo  {journal} {arXiv preprint arXiv:2209.08535}\
  } (\bibinfo {year} {2022}{\natexlab{b}})}\BibitemShut {NoStop}%
\bibitem [{\citenamefont {Lemmens}\ and\ \citenamefont
  {Seidel}(1973)}]{LemmensSeidel73}%
  \BibitemOpen
  \bibfield  {author} {\bibinfo {author} {\bibfnamefont {P.}~\bibnamefont
  {Lemmens}}\ and\ \bibinfo {author} {\bibfnamefont {J.}~\bibnamefont
  {Seidel}},\ }\href {https://doi.org/10.1016/0021-8693(73)90123-3} {\bibfield
  {journal} {\bibinfo  {journal} {Journal of Algebra}\ }\textbf {\bibinfo
  {volume} {24}},\ \bibinfo {pages} {494} (\bibinfo {year} {1973})}\BibitemShut
  {NoStop}%
\bibitem [{\citenamefont {Lin}\ and\ \citenamefont {Yu}(2020)}]{LinYu2020}%
  \BibitemOpen
  \bibfield  {author} {\bibinfo {author} {\bibfnamefont {Y.-C.~R.}\
  \bibnamefont {Lin}}\ and\ \bibinfo {author} {\bibfnamefont {W.-H.}\
  \bibnamefont {Yu}},\ }\bibfield  {journal} {\bibinfo  {journal} {Discrete
  Math.}\ }\textbf {\bibinfo {volume} {343}},\ \href
  {https://doi.org/10.1016/j.disc.2019.111667} {10.1016/j.disc.2019.111667}
  (\bibinfo {year} {2020})\BibitemShut {NoStop}%
\bibitem [{\citenamefont {Haantjes}(1948)}]{haantjes1948equilateral}%
  \BibitemOpen
  \bibfield  {author} {\bibinfo {author} {\bibfnamefont {J.}~\bibnamefont
  {Haantjes}},\ }\href@noop {} {\bibfield  {journal} {\bibinfo  {journal}
  {Nieuw Arch. Wiskunde (2)}\ }\textbf {\bibinfo {volume} {22}},\ \bibinfo
  {pages} {355} (\bibinfo {year} {1948})}\BibitemShut {NoStop}%
\bibitem [{\citenamefont {Renes}\ \emph {et~al.}(2004)\citenamefont {Renes},
  \citenamefont {Blume-Kohout}, \citenamefont {Scott},\ and\ \citenamefont
  {Caves}}]{Renes2004}%
  \BibitemOpen
  \bibfield  {author} {\bibinfo {author} {\bibfnamefont {J.~M.}\ \bibnamefont
  {Renes}}, \bibinfo {author} {\bibfnamefont {R.}~\bibnamefont {Blume-Kohout}},
  \bibinfo {author} {\bibfnamefont {A.~J.}\ \bibnamefont {Scott}},\ and\
  \bibinfo {author} {\bibfnamefont {C.~M.}\ \bibnamefont {Caves}},\ }\href
  {https://doi.org/10.1063/1.1737053} {\bibfield  {journal} {\bibinfo
  {journal} {Journal of Mathematical Physics}\ }\textbf {\bibinfo {volume}
  {45}},\ \bibinfo {pages} {2171} (\bibinfo {year} {2004})}\BibitemShut
  {NoStop}%
\bibitem [{\citenamefont {Scott}\ and\ \citenamefont
  {Grassl}(2010)}]{Scott2010}%
  \BibitemOpen
  \bibfield  {author} {\bibinfo {author} {\bibfnamefont {A.~J.}\ \bibnamefont
  {Scott}}\ and\ \bibinfo {author} {\bibfnamefont {M.}~\bibnamefont {Grassl}},\
  }\href {https://doi.org/10.1063/1.3374022} {\bibfield  {journal} {\bibinfo
  {journal} {Journal of Mathematical Physics}\ }\textbf {\bibinfo {volume}
  {51}},\ \bibinfo {pages} {042203} (\bibinfo {year} {2010})}\BibitemShut
  {NoStop}%
\bibitem [{\citenamefont {Wharton}\ and\ \citenamefont
  {Koch}(2015)}]{wharton2015unit}%
  \BibitemOpen
  \bibfield  {author} {\bibinfo {author} {\bibfnamefont {K.}~\bibnamefont
  {Wharton}}\ and\ \bibinfo {author} {\bibfnamefont {D.}~\bibnamefont {Koch}},\
  }\href@noop {} {\bibfield  {journal} {\bibinfo  {journal} {Journal of Physics
  A: Mathematical and Theoretical}\ }\textbf {\bibinfo {volume} {48}},\
  \bibinfo {pages} {235302} (\bibinfo {year} {2015})}\BibitemShut {NoStop}%
\bibitem [{\citenamefont {Penrose}(1955)}]{penrose1955generalized}%
  \BibitemOpen
  \bibfield  {author} {\bibinfo {author} {\bibfnamefont {R.}~\bibnamefont
  {Penrose}},\ }\href@noop {} {\bibfield  {journal} {\bibinfo  {journal}
  {Mathematical proceedings of the Cambridge philosophical society}\ }\textbf
  {\bibinfo {volume} {51}},\ \bibinfo {pages} {406} (\bibinfo {year}
  {1955})}\BibitemShut {NoStop}%
\bibitem [{\citenamefont {Don}(1985)}]{don1985use}%
  \BibitemOpen
  \bibfield  {author} {\bibinfo {author} {\bibfnamefont {F.~H.}\ \bibnamefont
  {Don}},\ }\href@noop {} {\bibfield  {journal} {\bibinfo  {journal} {Linear
  Algebra and its Applications}\ }\textbf {\bibinfo {volume} {70}},\ \bibinfo
  {pages} {225} (\bibinfo {year} {1985})}\BibitemShut {NoStop}%
\bibitem [{\citenamefont {Kato}(2013)}]{kato2013perturbation}%
  \BibitemOpen
  \bibfield  {author} {\bibinfo {author} {\bibfnamefont {T.}~\bibnamefont
  {Kato}},\ }\href@noop {} {\emph {\bibinfo {title} {Perturbation theory for
  linear operators}}},\ Vol.\ \bibinfo {volume} {132}\ (\bibinfo  {publisher}
  {Springer Science \& Business Media},\ \bibinfo {year} {2013})\BibitemShut
  {NoStop}%
\bibitem [{\citenamefont {Lemmens}\ \emph {et~al.}(1991)\citenamefont
  {Lemmens}, \citenamefont {Seidel},\ and\ \citenamefont
  {Green}}]{lemmens1991equiangular}%
  \BibitemOpen
  \bibfield  {author} {\bibinfo {author} {\bibfnamefont {P.~W.}\ \bibnamefont
  {Lemmens}}, \bibinfo {author} {\bibfnamefont {J.~J.}\ \bibnamefont
  {Seidel}},\ and\ \bibinfo {author} {\bibfnamefont {J.}~\bibnamefont
  {Green}},\ }in\ \href@noop {} {\emph {\bibinfo {booktitle} {Geometry and
  Combinatorics}}}\ (\bibinfo  {publisher} {Elsevier},\ \bibinfo {year}
  {1991})\ pp.\ \bibinfo {pages} {127--145}\BibitemShut {NoStop}%
\bibitem [{\citenamefont {Greaves}\ \emph {et~al.}(2016)\citenamefont
  {Greaves}, \citenamefont {Koolen}, \citenamefont {Munemasa},\ and\
  \citenamefont {Sz{\"o}ll{\H{o}}si}}]{greaves2016equiangular}%
  \BibitemOpen
  \bibfield  {author} {\bibinfo {author} {\bibfnamefont {G.}~\bibnamefont
  {Greaves}}, \bibinfo {author} {\bibfnamefont {J.~H.}\ \bibnamefont {Koolen}},
  \bibinfo {author} {\bibfnamefont {A.}~\bibnamefont {Munemasa}},\ and\
  \bibinfo {author} {\bibfnamefont {F.}~\bibnamefont {Sz{\"o}ll{\H{o}}si}},\
  }\href@noop {} {\bibfield  {journal} {\bibinfo  {journal} {Journal of
  Combinatorial Theory, Series A}\ }\textbf {\bibinfo {volume} {138}},\
  \bibinfo {pages} {208} (\bibinfo {year} {2016})}\BibitemShut {NoStop}%
\bibitem [{\citenamefont {Jiang}\ \emph {et~al.}(2021)\citenamefont {Jiang},
  \citenamefont {Tidor}, \citenamefont {Yao}, \citenamefont {Zhang},\ and\
  \citenamefont {Zhao}}]{jiang2021equiangular}%
  \BibitemOpen
  \bibfield  {author} {\bibinfo {author} {\bibfnamefont {Z.}~\bibnamefont
  {Jiang}}, \bibinfo {author} {\bibfnamefont {J.}~\bibnamefont {Tidor}},
  \bibinfo {author} {\bibfnamefont {Y.}~\bibnamefont {Yao}}, \bibinfo {author}
  {\bibfnamefont {S.}~\bibnamefont {Zhang}},\ and\ \bibinfo {author}
  {\bibfnamefont {Y.}~\bibnamefont {Zhao}},\ }\href@noop {} {\bibfield
  {journal} {\bibinfo  {journal} {Annals of Mathematics}\ }\textbf {\bibinfo
  {volume} {194}},\ \bibinfo {pages} {729} (\bibinfo {year}
  {2021})}\BibitemShut {NoStop}%
\bibitem [{\citenamefont {Godsil}\ and\ \citenamefont
  {Royle}(2001)}]{godsil01}%
  \BibitemOpen
  \bibfield  {author} {\bibinfo {author} {\bibfnamefont {C.}~\bibnamefont
  {Godsil}}\ and\ \bibinfo {author} {\bibfnamefont {G.~F.}\ \bibnamefont
  {Royle}},\ }\href {https://doi.org/10.1007/978-1-4613-0163-9} {\emph
  {\bibinfo {title} {Algebraic Graph Theory}}},\ \bibinfo {series} {Graduate
  Texts in Mathematics}\ No.\ \bibinfo {number} {Book 207}\ (\bibinfo
  {publisher} {Springer},\ \bibinfo {year} {2001})\BibitemShut {NoStop}%
\bibitem [{\citenamefont {Bravyi}\ \emph {et~al.}(2017)\citenamefont {Bravyi},
  \citenamefont {Gambetta}, \citenamefont {Mezzacapo},\ and\ \citenamefont
  {Temme}}]{bravyi2017tapering}%
  \BibitemOpen
  \bibfield  {author} {\bibinfo {author} {\bibfnamefont {S.}~\bibnamefont
  {Bravyi}}, \bibinfo {author} {\bibfnamefont {J.~M.}\ \bibnamefont
  {Gambetta}}, \bibinfo {author} {\bibfnamefont {A.}~\bibnamefont
  {Mezzacapo}},\ and\ \bibinfo {author} {\bibfnamefont {K.}~\bibnamefont
  {Temme}},\ }\href@noop {} {\bibfield  {journal} {\bibinfo  {journal} {arXiv
  preprint arXiv:1701.08213}\ } (\bibinfo {year} {2017})}\BibitemShut {NoStop}%
\bibitem [{\citenamefont {Bonechi}\ \emph {et~al.}(1992)\citenamefont
  {Bonechi}, \citenamefont {Celeghini}, \citenamefont {Giachetti},
  \citenamefont {Sorace},\ and\ \citenamefont
  {Tarlini}}]{bonechi1992heisenberg}%
  \BibitemOpen
  \bibfield  {author} {\bibinfo {author} {\bibfnamefont {F.}~\bibnamefont
  {Bonechi}}, \bibinfo {author} {\bibfnamefont {E.}~\bibnamefont {Celeghini}},
  \bibinfo {author} {\bibfnamefont {R.}~\bibnamefont {Giachetti}}, \bibinfo
  {author} {\bibfnamefont {E.}~\bibnamefont {Sorace}},\ and\ \bibinfo {author}
  {\bibfnamefont {M.}~\bibnamefont {Tarlini}},\ }\href@noop {} {\bibfield
  {journal} {\bibinfo  {journal} {Journal of Physics A: Mathematical and
  General}\ }\textbf {\bibinfo {volume} {25}},\ \bibinfo {pages} {L939}
  (\bibinfo {year} {1992})}\BibitemShut {NoStop}%
\bibitem [{\citenamefont {Cha}\ and\ \citenamefont
  {Shin}(2018)}]{cha2018perturbation}%
  \BibitemOpen
  \bibfield  {author} {\bibinfo {author} {\bibfnamefont {P.~D.}\ \bibnamefont
  {Cha}}\ and\ \bibinfo {author} {\bibfnamefont {A.}~\bibnamefont {Shin}},\
  }\href@noop {} {\bibfield  {journal} {\bibinfo  {journal} {Shock and
  Vibration}\ }\textbf {\bibinfo {volume} {2018}} (\bibinfo {year}
  {2018})}\BibitemShut {NoStop}%
\end{thebibliography}%
